\documentclass[10pt,letterpaper]{article}
\usepackage[margin=1in,nohead,centering]{geometry}
\usepackage{graphicx,color,wrapfig}
\usepackage{amsmath,amsfonts,amsthm}
\usepackage{textcomp}  

\usepackage{url}
\usepackage[linewidth=1pt]{mdframed}
\usepackage{lipsum}
\usepackage{xspace}
\usepackage{graphics}
\usepackage{caption}
\usepackage{capt-of} 
\usepackage{float} 
\usepackage{subcaption}

\newcommand{\ob}{\,\mbox{\bf\texttt{[}}\,}
\newcommand{\cb}{\,\mbox{\bf\texttt{]}}\,}
\newcommand{\op}{\,\mbox{\bf\texttt{(}}\,}
\newcommand{\cp}{\,\mbox{\bf\texttt{)}}\,}

\newtheorem{theorem}{Theorem}

\newtheorem{lemma}[theorem]{Lemma}
\newtheorem{corollary}[theorem]{Corollary}

\usepackage{textcomp}  

\newcommand{\ZZ}{{{\widehat{Z}}}\xspace}

\chardef\other=12

\def\mmakeactive#1{\catcode`#1=\active\ignorespaces}

{
\mmakeactive\
\gdef\obeywhitespace{%
  \mmakeactive\^^M %
  \let^^M=\NewLine %
  \aftergroup\removebox %
  \obeyspaces %
}}

\def\NewLine{\par\indent}
\def\removebox{\setbox0=\lastbox}

\def\|{|}

\title{On the scale-free nature of RNA secondary structure networks}
\author{P. Clote}
\date{Biology Department, Boston College, Chestnut Hill, MA 02467}

\begin{document}
\maketitle

\abstract{A network is {\em scale-free} if its connectivity density function
is proportional to
a power-law distribution.  Scale-free networks may provide an explanation for
the robustness observed in certain physical and biological phenomena, since the
presence of a few highly connected {\em hub} nodes and a large number of 
small-degree nodes may provide alternate paths between any two nodes on 
average -- such robustness has been suggested in studies of
metabolic networks, gene interaction networks and protein folding. A theoretical
justification for why biological networks are often found to be scale-free may lie
in the well-known fact that expanding networks in which new nodes are preferentially 
attached to highly connected nodes tend to be scale-free.
In this paper, we provide the first efficient algorithm to compute the
connectivity density function for the ensemble of all secondary structures of
a user-specified length, and show both by computational and theoretical arguments
that preferential attachment holds when expanding the network from length $n$ to
length $n+1$ structures.
Since existent power-law fitting software, such as {\tt powerlaw}, cannot
be used to determine a power-law fit for our exponentially large RNA connectivity data,
we also implement efficient code to compute the maximum
likelihood estimate for the power-law scaling factor and associated
Kolmogorov-Smirnov $p$-value.
Statistical goodness-of-fit tests indicate that one must reject the hypothesis that
RNA connectivity data follows a power-law distribution. Nevertheless, the power-law
fit is visually a good approximation for the tail of connectivity data,
and provides a rationale for investigation of preferential
attachment in the context of macromolecular folding.
}

\section{Introduction}
\label{section:intro}

The {\em connectivity} (or {\em degree}) of a node $v$ in a network 
(or undirected graph) is the number of nodes (or neighbors) of $s$,
connected to $s$ by an edge.
A network is said to be {\em scale-free} if
its connectivity function $N(k)$, which  represents the
number of nodes having degree $k$, satisfies the property that
$N(a\cdot k) = b\cdot N(x)$, the unique solution of which is
a {\em power-law} distribution, which by definition satisfies
$N(k) \propto {k^{-\alpha}}$ for some scaling factor $\alpha>1$  
\cite{newmanSIAMreview}. Scale-free
networks contain a few nodes of high degree and a large number of 
nodes of small degree, hence may provide a reasonable model to explain
the robustness often manifested in biological networks -- such robustness
must, of course, be present for life to exist.

Barab{\'a}si and Albert \cite{Barabasi.s99} analyzed the emergence of
scaling in random networks, and showed that two properties, previously 
not considered in graph theory, were responsible for the power-law 
scaling observed in real networks: (1) networks are not static, but grow
over time, (2) during network growth, a highly connected node tends to 
acquire even more connections -- the latter concept is known as
{\em preferential attachment}.
In \cite{Barabasi.s99}, it was argued that
preferential attachment of new nodes implies that the degree
$N(k)$ with which a node in the network interacts with $k$ other nodes decays 
as a power-law, following $N(k) \propto k^{-\alpha}$, for 
$\alpha>1$.  This argument provides a plausible explanation for why
diverse biological and physical networks appear to be scale-free.
Indeed, various publications have suggested that the the following
biological networks are scale-free: protein-protein
interaction networks \cite{Ito.pnas00,schwikowskiProteinProteinNetworks},
metabolic networks \cite{Ma.b03}, 
gene interaction networks \cite{Tong.s04}, 
yeast co-expression networks \cite{VanNoort.er04},
and protein folding networks \cite{Bowman.pnas10}.

\subsubsection*{How scale-free are biological networks?}

The validity of a power-law fit for previously studied
biological networks was first called into question in \cite{Khanin.jcb06}, 
where 10 published data sets of biological interaction networks were shown
{\em not} to be fit by a power-law distribution, despite published claims 
to the contrary. Estimating an optimal power-law scaling factor by maximum
likelihood and using $\chi^2$ goodness-of-fit tests, it was shown in
\cite{Khanin.jcb06} that not a single interaction network from
had a nonzero probability of being drawn from a power-law distribution;
nevertheless, some of the interaction networks could be fit by a
{\em truncated} power-law distribution.
The data analyzed by the authors included data from protein-protein 
interaction networks \cite{Ito.pnas00,schwikowskiProteinProteinNetworks},
gene interaction networks determined by synthetic lethal interactions
\cite{Tong.s04}, metabolic interaction networks \cite{Ma.b03}, etc. 

In \cite{newmanPowerLawFitEmpiricalData}, 24 real-world data sets were
analyzed from a variety of disciplines, each of which had been conjectured 
to follow a power-law
distribution. Estimating an optimal power-law scaling factor by maximum
likelihood and using goodness-of-fit tests based on likelihood ratios and
on the Kolmogorov-Smirnov statistic for non-normal data, it was shown in
\cite{newmanPowerLawFitEmpiricalData} that some of the conjectured
power-law distributions were consistent with claims in the literature, while
others were not. For instance, 
Clauset et al.  \cite{newmanPowerLawFitEmpiricalData} found 
sufficient statistical evidence to reject claims of scale-free behavior for
earthquake intensity and metabolic degree networks, while there was 
insufficient evidence to reject such claims for
networks of protein interaction, Internet, and species per genus.

It is possible to come to opposite conclusions, depending on whether $\chi^2$ 
or Kolmogorov-Smirnov (KS) statistics are used to test the hypothesis whether
a network is scale-free, i.e. follows a (possibly truncated) power-law
distribution.  Indeed, Khanin and Wit \cite{Khanin.jcb06} obtained a $p$-value 
of $<10^{-4}$ for $\chi^2$ goodness-of-fit for a truncated power-law
distribution for the protein-protein interaction data from
\cite{Ito.pnas00}, while
Clauset et al.  \cite{newmanPowerLawFitEmpiricalData} obtained a
$p$-value of $0.31$ for KS goodness-of-fit for a truncated power-law
for the same data. This example provides the occasion for us to explain the position
taken in this paper that (in our opinion) it is quite possible for a statistical
test to lead to the rejection of goodness-of-fit of the power-law distribution for
physical data arising from biological networks, yet the (approximate) power-law
fit can possibly provide valuable insight into the nature of the data. In this
manner, we sidestep the current polemic concerning the question of how wide-spread 
scale-free networks really are. In their preprint from Jan. 9, 2018,
entitled ``Scale-free networks are rare'', Broido and Clauset 
\cite{clausetArxiv2018} argue that less than
45 of the 927 real-world network data sets (i.e. $4\%$) found in the
{\em Index of Complex Networks} exhibit the
``strongest level of direct evidence for scale-free structure''. In
a response statement dated March 6, 2018, A.L. Barab{\'a}si argued against
the conclusions of Broido and Clauset -- indeed, the title of 
Barab{\'a}si's statement sums up his position:
``Love is All You Need:
{C}lauset's fruitless search for scale-free networks''.\footnote{It is not
the first time a polemic has arisen concerning the
power-law distribution -- indeed, there was a heated exchange
between Mandelbrot and Simon almost 60 years ago in the journal 
{\em Information and Control}. For details, references, and a history
of the power-law distribution, see see \cite{mitzenmacher}.}

Regardless of the Barab{\'a}si-Clauset polemic,
we stress that prior to the introduction of our novel secondary structure
connectivity algorithm, only fragmentary results were possible 
by exhaustively enumerating all secondary structures having free energy
within a certain range obove the minimum free energy \cite{Wuchty.nar03}.
Indeed, using our methods, for the first time we can address the question of
whether RNA secondary structure connectivity is scale-free.
Moreover, to the best of our knowledge, this is the first time that any computational
or theoretical evidence has been given to show that preferential attachment
exists for the network of RNA secondary structures. 

The current paper investigates properties of the ensemble of RNA secondary structures,
considered as a network, and so extends results of \cite{Clote.jcc15},
which described a cubic time dynamic programming algorithm to compute the expected network
degree. The RNA connectivity algorithm described in 
Section~\ref{section:fasterAlgoHomopolymer} is completely unrelated from  that of
\cite{Clote.jcc15}, yet allows one to compute all finite moments,
including mean, variance, skew, etc.

The plan of the remaining paper is as follows.
Section~\ref{section:methods} presents a brief
summary of basic definitions, followed by the recursions for
an efficient dynamic programming algorithm
to determine the absolute [resp. relative] frequencies
$N(k)$ [resp. $p(k)$ for secondary structure connectivity
of a given homopolymer, which allows non-canonical base pairs. 
Though not done in this paper, this algorithm could be extended
to the case of (real) RNA sequences allowing only Watson-Crick
and wobble base pairs.  Section~\ref{section:statisticalMethods}
presents the statistical methods used to both fit RNA connnectivity
data to a power-law distribution and to perform a goodness-of-fit
test using Kolmogorov-Smirnov distance.
Section~\ref{section:results} presents results on power-law
fits of RNA connectivity data, and computational evidence that
preferential attachment holds for RNA secondary structure networks.
Section~\ref{section:conclusion} presents concluding remarks, and the Appendix
presents a mathematical proof of preferential attachment in the case of
a simplified model of secondary structure.

\section{Computing degree frequency}
\label{section:methods}

Section~\ref{section:preliminaries} presents basic definitions and notation used;
Section~\ref{section:degreeDensityAlgorithm} presents an algorithm to compute the
frequency of each degree less than $K$ in the ensemble of all secondary structures
with run time $O(K^2 n^4)$ and memory requirements $O(K n^3)$. 
Section~\ref{section:fasterAlgoHomopolymer} presents a more efficient algorithm,
with run time $O(K^2 n^3)$ and memory requirements $O(K n^2)$, for the special case
of a homopolymer, in which all possible non-canonical base pairs are permitted.
We implemented both algorithms in Python, cross-checked for identical results,
and call the resulting code {\tt RNAdensity}.
Since this paper is a theoretical contribution on network properties, we focus only
on homopolymers and do not present the details necessary to extend the algorithm of
Section~\ref{section:degreeDensityAlgorithm} to non-homopolymer RNA, where base pairs
are required to be Watson-Crick or GU wobble pairs.

\subsection{Preliminaries}
\label{section:preliminaries}

A secondary structure for a length $n$ {\em homopolymer} is a set $s$ of
base pairs $(i,j)$, such that (1) there exist at least $\theta$ unpaired
bases in every hairpin, where $\theta$ is usually taken to be $3$, though
sometimes $1$ in the literature, (2) there do not exist base pairs
$(i,j), (k,\ell) \in s$, such that $i<k<j<\ell$; i.e. a secondary structure
is an outerplanar graph, where each base pair $(i,j) \in s$ satisfies
$j-i>\theta$. The {\em free energy} of a homopolymer secondary structure $s$
is defined to be $-1$ times the number $|s|$ of base pairs in $s$ 
(Nussinov-Jacobson energy model \cite{nussinovJacobson}). Since
entropic effects are ignored, this is not a real free energy; however it
allows us to use the standard notation ``MFE'' for `minimum free energy'.
Note that the MFE structure for a length $n$ homopolymer has
$\lfloor \frac{n-\theta}{2} \rfloor$ many base pairs.

For a given RNA sequence, consider the exponentially large network of all
its secondary structures, where an undirected edge exists between any two
structures $s$ and $t$, whose base-pair distance equals one -- in other
words, for which $t$ is obtained from $s$ by either removing or adding 
one base pair. The connectivity (or degree) of a node, or structure, $s$ is 
defined to be the number of secondary structures obtained by deleting
or adding one base pair to $s$ -- this corresponds to the so-called
$MS_1$ move set \cite{flamm}.  At the end of the paper, we
briefly consider the $MS_2$ move set, where the degree of a structure $s$
is defined to be the number of secondary structures obtained by adding,
deleting or {\em shifting} one base pair \cite{amirShift}. The
$MS_1$ [resp. $MS_2$] connectivity of the MFE structure for a homopolymer
of length $n$ is $\lfloor \frac{n-\theta}{2} \rfloor$
[resp. $\lceil \frac{n-\theta}{2} \rceil$]. 
{\em Connectivity} $N(k)$ is defined to be the {\em absolute} frequency of 
degree $k$, i.e. the number of secondary structures having exactly $k$ 
neighbors, that can be obtained by either adding or removing a single
base pair.  The {\em degree density} $p(k)$ is defined to be the
probability density function (PDF) or
{\em relative} frequency of $k$, i.e. the proportion $p(k) = \frac{N(k)}{Z}$
of all secondary structures having $k$ neighbors, where $Z$ denotes the
total number of secondary structures for a given homopolymer.
A network is defined to be {\em scale-free},
provided its degree frequency $N(k)$ is proportional to a power-law, i.e.
$N(k) \propto k^{-\alpha}$ where $\alpha>1$ is the {\em scaling factor}.

\subsection{Computing the degree density}
\label{section:degreeDensityAlgorithm}

In this section, we describe a novel dynamic programming algorithm to 
compute the $MS_1$ {\em degree density} $p(k)$ for the network of secondary 
structures for a homopolymer of length $n$.  Note first that the empty structure
$s_{\emptyset}$ of length $n$ has 
\begin{align}
\label{eqn:degreeEmptyStr}
\mbox{degree}(s_{\emptyset}) &= \frac{(n-\theta)(n-\theta-1)}{2}
\end{align}
many neighbors, each obtained by adding a base pair. Indeed,
\begin{align*}
\mbox{degree}(s_{\emptyset}) &= 
\sum_{i=1}^{n-\theta-1} \sum_{j=i+\theta+1}^n 1 =
\sum_{i=1}^{n-\theta-1} [n-(i+\theta+1)+1] \\
&=
\sum_{i=1}^{n-\theta-1} (n-i-\theta)
= (n-\theta)(n-\theta-1) - \sum_{i=1}^{n-\theta-1} i =
\frac{(n-\theta)(n-\theta-1)}{2}
\end{align*}
Using a simple induction argument, equation~(\ref{eqn:degreeEmptyStr})
implies that for all values of $n$, the maximum possible degree,
$\mbox{maxDegree}(n)$, of a secondary
structure for the length $n$ homopolymer is
$frac{(n-\theta)(n-\theta-1)}{2}$

Let $Z^*(i,j,k)$ denote the
number of secondary structures on the interval $[i,j]$ that have exactly
$k$ neighbors with respect to the $MS_1$ move set (i.e. have degree $k$).
Let $N(i,j)$ denote the number of secondary structures on interval $[i,j]$,
computed by simple recursions from \cite{steinWaterman}
\begin{align}
\label{eqn:numStr1}
N(i,j) &= \left\{ \begin{array}{ll}
1 &\mbox{if $1\leq i\leq j \leq i+\theta \leq n$}\\
N(i,j-1) + N(i+1,j-1) + \sum_{r=i+1}^{j-\theta-1}   N(i,r-1) \cdot N(r+1,j-1)
&\mbox{if $i+\theta+1 \leq j \leq n$}\\
\end{array} \right.
\end{align}
or more simply
\begin{align}
\label{eqn:numStr2}
N(m) &= \left\{ \begin{array}{ll}
1 &\mbox{if $1 \leq m \leq \theta+1$}\\
N(m-1) + N(m-2) + \displaystyle\sum_{r=\theta}^{m-3}   N(m-r-2) \cdot N(r)
&\mbox{if $\theta+2 \leq m \leq n$}
\end{array} \right.
\end{align}
Although recursion equation (\ref{eqn:numStr1}) requires $O(n^3)$ time and 
$O(n^2)$ space, it can trivially be extended to compute the number of secondary
structures for an arbitary RNA sequence $a_1,\ldots,a_n$, where base pairs
are either Watson-Crick or wobble pairs. If no such extension is necessary,
then the recursion equation (\ref{eqn:numStr2}), first given in 
\cite{steinWaterman}, requires $O(n^2)$ time and $O(n)$ space, hence is
more efficient by a factor of $n$. In a similar fashion, the recursion
equations (\ref{eqn:baseA}-\ref{eqn:indD}) and pseudocode in
Section~\ref{section:degreeDensityAlgorithm} are given in a form that allows an
extension (not given here) to the general case of computing the 
degree density for
the ensemble of secondary structures of a given RNA sequence $a_1,\ldots,a_n$.
The resulting code ref{algo:degreeDensity} requires
$O(n^6)$ time and $O(n^4)$ storage, but this can be improved by a factor
of $n$.

Suppose that every hairpin loop is required to have at least 
$\theta \geq 1$ unpaired positions; i.e. if $(i,j)$ is a base pair,
then $i+\theta+1 \leq j$.  As in the 
recursions (\ref{eqn:baseA}-\ref{eqn:indD}), let
$Z(i,j,k,h,v)$ denote the number of secondary structures on 
the interval $[i,j]$, for $1 \leq i \leq j \leq n$ for the homopolymer model,
that have exactly $k$ neighbors, and for which there are exactly $h$ 
unpaired positions (or {\em holes}) in $[i,j-\theta-1]$ and the position 
$j-v$ is paired to $r \in [i,j-v-\theta-1]$, 
while positions $j-v,j-v+1,...,j$ are not base-paired to any position in 
$[i,j]$. Additionally, define 
\begin{align}
\label{eqn:Zstar}
Z^*(i,j,k) = \sum_{h=0}^{j-\theta-i} \sum_{v=0}^{\theta+1} Z(i,j,k,h,v)
\end{align}
Recalling from equation (\ref{eqn:degreeEmptyStr}) that
$\mbox{maxDegree}(n) = \frac{(n-\theta)(n-\theta-1)}{2}$,
for any $1 \leq i\leq j \leq n$, we clearly have that
\begin{align*}
N(i,j) &= \sum_{k=1}^{\mbox{\tiny{maxDegree(j-i+1)}}} Z^*(i,j,k) \\
 &= \sum_{k=1}^{\mbox{\tiny{maxDegree(j-i+1)}}}
 \sum_{h=0}^{j-\theta-i} \sum_{v=0}^{\theta+1} Z(i,j,k,h,v)
\end{align*}


The idea of our algorithm is to partition all secondary structure of the
interval $[i,j]$ into those structures having exactly degree $k$ 
($k$ $MS_1$ neighbors, i.e. $k$ structures that can be obtained by either
adding or removing a single base pair). To support an inductive argument,
in proceeding from interval $[i,j]$ to $[i,j+1]$, we need additionally to
determine the number of structures having degree $k$, which have a certain
number $h$ of positions that are {\em visible} (external to every base pair),
which can be paired with the last position $j+1$. Note that the position
$j-\theta$ can {\em not} be base-paired with $j$ in $[i,j]$; however,
$j-\theta$ {\em can} be base-paired with $j$ in $[i,j+1]$. Thus in addition
to keeping track of the number $h$ of {\em holes} (positions in 
$i,\ldots,j-\theta-1$ that are external to all base pairs, hence can be paired
with $j$), we introduce the variable $v$ to keep track of the number of 
{\em visible} positions in $j-\theta,\ldots,j$. This explains our need for
the function $Z(i,j,k,h,v)$ as defined in 
equations (\ref{eqn:baseA}-\ref{eqn:indD}). We now proceed to the details,
where for ease of the reader, some definitions are repeated.

Let $\theta=3$ denote the minimum number of unpaired positions required to
be present in a hairpin loop. For a length $n$ homopolymer,
let $1 \leq i \leq j \leq n$, 
$0 \leq k \leq {{n-\theta} \choose 2}$, $0 \leq h \leq j-i-\theta$,
$0 \leq v \leq \theta+1$. Recall that
$Z(i,j,k,h,v)$ denotes the number of secondary structures on
$[i,j]$ for the homopolymer model, that have exactly $k$ $MS_1$ neighbors
(i.e. degree $k$), and there are exactly $h$ unpaired positions
in $[i,j-\theta-1]$ and the position $j-v$ is base-paired to some
$r \in [i,j-v-\theta-1]$ while positions $j-v,j-v+1,\ldots,j$
are not base-paired to any position in $[i,j]$. The parameter $h$
corresponds to the number of {\em visible positions} or {\em holes} 
$[i,j-\theta-1]$ that are external to base pairs in $[i,j]$,
while the parameter $v$ 
corresponds to the number of {\em visible} positions in 
$[j-\theta,j]$ that are external to base pairs in $[i,j]$.

Recall our notation $Z^*(i,j,k) = \sum_h \sum_v Z(i,j,k,h,v)$. 
We begin by initializing $Z(i,j,k,h,v)=0$ for all values in corresponding
ranges. Letting $N(i,j)$ denote the number of secondary structures on
$[i,j]$ for the homopolymer model, as computed by equation (\ref{eqn:numStr1}),
the following recursions describe an algorithm that requires
$O(K \cdot n^3)$ storage and $O(K^2 \cdot n^4)$ time to compute the
probability $Prob[ \mbox{\tiny deg}(s) = k] = \frac{Z^*(1,n,k)}{N(1,n)}$
that a (uniformly chosen) random secondary structure has degree $k$ for
$0 \leq k \leq K$, where $K$ is a user-defined constant bounded above by
$\mbox{maxDegree}(n) = \frac{(n-\theta)(n-\theta-1)}{2}$.
\medskip

Base Case A considers all structures on $[i,j]$, as depicted in
Figure~\ref{fig:baseA}, that are too small to
have any base pairs, hence which have degree zero.

\noindent
{\bf Base Case A:} For $j-i \leq \theta$, define \hfill\break
\begin{align}
\label{eqn:baseA}
Z(i,j,0,0,j-i+1) &= 1
\end{align}

\begin{center} 
\includegraphics[width=0.15\textwidth]{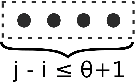}
\captionof{figure}[]{Structures considered in base case A.}
\label{fig:baseA}
\end{center}
\medskip

Base Case B considers all structures on $[i,j]$, as depicted in
Figure~\ref{fig:baseB}, that have only base pair $(i,j)$, since other
potential base pairs would contain fewer than $\theta$ unpaired bases.
The degree of such structures is $1$, since only one base pair can be
removed, and no base pairs can be added.
Moreover, no position in $[i,j]$ is external to the base pair $(i,j)$, so
visibility parameters $h=0,v=0$. The arrow in
Figure~\ref{fig:baseB} indicates that the sole neighbor is the empty
structure, obtained by removing the base pair $(i,j)$.

\noindent
{\bf Base Case B:} For $j-i = \theta+1$ and
$(i,j)$ is a base pair, define  \hfill\break
\begin{align}
\label{eqn:baseB}
Z(i,j,1,0,0) &=1
\end{align}

\begin{center} 
\includegraphics[width=0.35\textwidth]{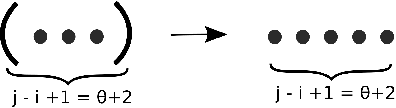}
\captionof{figure}[]{Structures considered in base case B.}
\label{fig:baseB}
\end{center}
\medskip

Base Case C considers the converse situation, consisting of the empty
structure on $[i,j]$ where $j-i = \theta+1$, whose sole neighbor is
the structure consisting of base pair $(i,j)$. The arrow is meant to
indicate that the structure on the right is the only neighbor of that on
the left, as depicted in Figure~\ref{fig:baseC}. Since the size of
the empty structure on $[i,j]$ is $\theta+2$ and every position in
$[i,j]$ is visible (external to every base pair), $h=1$ and $v=\theta+1$. 
the dotted rectangle in Figure~\ref{fig:baseC} indicates the $\theta+1$
unpaired positions at the right extremity as counted by $v=\theta+1$.

\noindent
{\bf Base Case C:} For $j-i = \theta+1$ and
$(i,j)$ not base-paired, define  \hfill\break
\begin{align}
\label{eqn:baseC}
Z(i,j,1,1,\theta+1) &=1
\end{align}
\medskip
\begin{center} 
\includegraphics[width=0.35\textwidth]{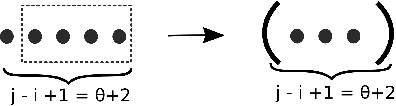}
\captionof{figure}[]{Structures considered in base case C.}
\label{fig:baseC}
\end{center}
\medskip

Base Case D considers the empty structure on $[i,j]$ where
$j-i>\theta+1$. The empty structure is the only structure having degree
maxDegree$(i,j) = \frac{(j-i-\theta+1)(j-i-\theta)}{2}$, since
maxDegree$(i,j)$ many base pairs can be added to the empty structure.
In Figure~\ref{fig:baseD}, the dotted rectangle indicates the
$\theta+1$ rightmost unpaired positions, corresponding to visibility
parameter $v=\theta+1$, while dotted circles indicate the $h = j-i-\theta$
holes, i.e. unpaired positions that could be paired with the rightmost position
$j$.

\noindent
{\bf Base Case D:} For all $(j-i+1) > \theta+2$, the empty structure,
as indicated by $h+v=j-i+1$ (so $h=j-i-\theta$ and $v=\theta+1$), 
has degree maxDegree$(i,j)$ as defined by equation 
\ref{eqn:degreeEmptyStr}, where
\hfill\break
\begin{align}
\label{eqn:baseD}
Z(i,j,\mbox{maxDegree}(i,j),j-i-\theta,\theta+1)  &=1
\end{align}
\medskip
\begin{center} 
\includegraphics[width=0.25\textwidth]{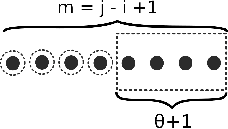}
\captionof{figure}[]{Structures considered in base case D.}
\label{fig:baseD}
\end{center}
\medskip

Inductive Case A considers the case where left and right extremities
$i,j$ form the base pair $(i,j)$, where $j-i>\theta+1$. No position
in $[i,j]$ is visible (external to all base pairs), so visibility
parameters $h=0=v$. Recalling the definition of $Z^*(i,j,k)$
from equation \ref{eqn:Zstar}, we have the following.

\noindent
{\bf Inductive Case A:} For $j-i > \theta+1$ and $(i,j)$ base-paired in 
$[i,j]$,
 \hfill\break
\begin{align}
\label{eqn:indA}
Z(i,j,k,0,0) &= Z(i,j,k,0,0) + Z^*(i+1,j-1,k-1)
\end{align}
From this point on, we use the operator $+=$, so that the previous equation
would be written as $Z(i,j,k,0,0) += Z^*(i+1,j-1,k-1)$.
\medskip
\begin{center} 
\includegraphics[width=0.15\textwidth]{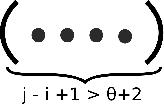}
\captionof{figure}[]{Structures considered in inductive case A.}
\label{fig:indA}
\end{center}
\medskip

Inductive Case B considers the case where last position $j$ base-pairs with
the $r$, where $i<r<j-\theta$. The value $r=i$ has already been considered
in Inductive Case A, and values $r=j-\theta+1,\ldots,j-1$ cannot base-pair to
$j$, since the corresponding hairpin loop would constain less than $\theta$ 
unpaired positions. This situation is depicted in Figure~\ref{fig:indB},
where there are $h$ holes (positions in $[i,j-\theta-1]$ that are external 
to all base pairs) and no visible positions in $[j-\theta,j]$.

\noindent
{\bf Inductive Case B:} For $j-i > \theta+1$ and $(r,j)$ base-paired 
in $[i,j]$ for some $i<r<j-\theta$,
 \hfill\break
\begin{align}
\label{eqn:indB}
Z(i,j,k,h,0) &+= \sum\limits_{r=i+1}^{j-\theta-1} \sum\limits_{k_1+k_2 = k-1}
\sum\limits_{w=0}^{\theta+1} Z(i,r-1,k_1,h-w,w) \cdot
Z^*(r+1,j-1,k_2)
\end{align}
When implemented, this requires a check that $h-w \geq 0$.
\medskip
\begin{center} 
\includegraphics[width=0.35\textwidth]{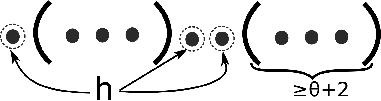}
\captionof{figure}[]{Structures considered in inductive case B.}
\label{fig:indB}
\end{center}
\medskip

For each value $v \in \{1,\ldots,\theta+1\}$,
inductive Case $C(v)$ considers the case where 
position $r \in [i,j-v-\theta-1]$ forms a base pair with position
$j-v$. The value $v=0$ is not considered here, since it was already
considered in Inductive Cases A,B. Note that a structure $s$ of the
format has $k$ neighbors, provided the restriction of $s$ to
$[i,r-1]$ has $k_1$ neighbors, and the restriction of $s$ to 
$[r+1,j-1]$ has $k_2$ neighbors, where $k_1+k_2+vh+1=k$. The term
$vh$ is due to the fact that since base pair $(r,j-v)$ ensures that
all {\em holes} are located in $[i,r-1]$, hence located at 
more than $\theta+1$ distance from all {\em visible}
positions in $[j-v+1,j]$, a neighbor of $s$ can be obtained by adding
a base pair from any hole to any visible suffix position -- there are
$vh$ many such possible base pairs that can be added. Finally, the
last term $+1$ is present, since one neighbor of $s$ can obtained by 
removing base pair $(r,j-v)$. This explains the summation indices and
summation terms in equation~(\ref{eqn:indCv}).
Figure \ref{fig:indCv} depicts a typical structure considered in
case $C(v)$.

\noindent
{\bf Inductive Case C($v$), for $v \in \{1,2,\ldots,\theta+1\}$:} 
For $j-i > \theta+1$ and $(r,j-v)$ base-paired in $[i,j]$,
for some $i<r<j-v-\theta$, where $j-v+1,\ldots,j$ are unpaired 
in $[i,j]$,
\hfill\break
\begin{align}
\label{eqn:indCv}
Z(i,j,k,h,v) &+= Z^*(2,j-1-v,k-1-vh) \\
\nonumber
&+
\sum\limits_{r=i+1}^{j-v-\theta-1}
\sum\limits_{k_1+k_2=(k-1-v h)} \qquad
\sum\limits_{w=0}^{\theta+1} Z(i,r-1,k_1,h-w,w) \cdot
Z^*(r+1,j-1-v,k_2)
\end{align}
The first term $Z^*(2,j-1-v,k-1-vh)$ handles the subcase where
$r=1$, so that $(1,j-v)$ is a base pair, while the second term handles the
subcase where $r>1$. Note that when implemented, this requires a test that
$h-w\geq 0$.

\medskip
\begin{center} 
\includegraphics[width=0.40\textwidth]{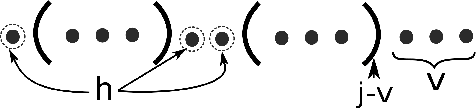}
\captionof{figure}[]{Structures considered in inductive case C$(v)$.}
\label{fig:indCv}
\end{center}
\medskip

Case $D$ considers the case where there are $h$ holes, and
positions $j-\theta-1,\ldots,j$ are unpaired, so that $v=\theta+1$. 
Note that $v=\theta+1$ implies only that $j-\theta,\ldots,j$ are unpaired,
so Case $D$ includes the addition requirement that position
$j-\theta-1$ is unpaired. Structures $s$ satisfying Case $D$ can be
partitioned into subcases where the restriction of $s$ to $[i,j-\theta-1]$
has $h-w$ holes in $[i,(j-\theta-1)-(\theta+1)] = [i,j-2\theta-2]$, and
$1 \leq w \leq \theta+1$ visible positions in $[j-2\theta-1,j-\theta-1]$.
Note that $(h-w)+w=h$, accounting for the $h$ holes in structure $s$
in $[i,j-\theta-1]$, and that it is essential that $w\geq 1$, since the
case $w=0$ was considered in Case $C(\theta+1)$.

The term $\frac{w(w+1)}{2}$ is due to the fact that the
rightmost position $j-\theta-1$ in the restriction of $s$ to $[i,j-\theta-1]$ 
can base-pair with position $j$, but not with $j-1$, etc. since this would
violate the requirement of at least $\theta$ unpaired bases in a hairpin
loop. Similarly, the second rightmost position $j-\theta-2$ in the 
restriction of $s$ to $[i,j-\theta-1]$ can base-pair with positions $j$
and $j-1$, but not with $j-2$, etc.; as well, the
third rightmost position $j-\theta-3$ can base-pair with positions $j$,
$j-1$ and $j-2$, but not with $j-3$, etc. The number of neighbors of
$s$ produced in this fashion is thus $\sum_{i=1}^w i = \frac{w(w+1)}{2}$.
Finally, the
term $(\theta+1)(h-w)$ is due to the fact that each of the 
$h-w$ holes in the restriction of $s$ to $[i,j-\theta-1]$ can base-pair
to each of the $(\theta+1)$ positions in $[j-\theta,j]$.

The argument just given shows the following. 
Let $s$ be a structure that satisfies
conditions of Case $D$ with $h$ holes and $v=\theta+1$ visible positions, and
suppose that the restriction of $s$ to $[i,j-\theta-1]$ has
$h-w$ holes and $w$ visible positions. Then
$s$ has $k$ neighbors provided that the restriction of $s$ to
$[i,j-\theta-1]$ has $k-\frac{w(w+1)}{2} - (\theta+1)(h-w)$ neighbors
on interval $[i,j-\theta-1]$. The equation (\ref{eqn:indD}) now follows.

\noindent
{\bf Inductive Case D:} 
For $j-i > \theta+1$ and 
$j-\theta-1,j-\theta,\ldots,j$ unpaired in $[i,j]$, and $1 \leq h < j-\theta-i$,
\hfill\break
\begin{align}
\label{eqn:indD}
Z(i,j,k,h,\theta+1) &+= \sum\limits_{w=1}^{\theta+1}
Z(i,j-\theta-1,k- \frac{w(w+1)}{2} - (\theta+1)\cdot(h-w), h-w,w)
\end{align}
\begin{center} 
\includegraphics[width=0.40\textwidth]{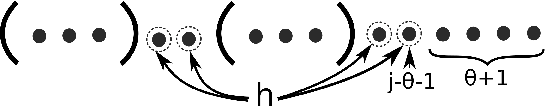}
\captionof{figure}[]{Structures considered in inductive case D.}
\end{center}
As in Case C($v$), when implemented, this requires a test that
$h-w\geq 0$.

Our implementation of the recursions (\ref{eqn:baseA}-\ref{eqn:indD})
has been cross-checked with exhaustive enumeration; moreover, we always have
that $\sum_{k} Z^*(i,j,k) = N(i,j)$, so the degree density is correctly
computed.

\subsection{Faster algorithm in the homopolymer case}
\label{section:fasterAlgoHomopolymer}

The algorithm described in Section~\ref{section:degreeDensityAlgorithm}
requires $O(K^2n^4)$ time and $O(Kn^3)$ space, where $K$ is a user-specified
degree bound $K \leq \frac{(n-\theta)(n-\theta-1)}{2}$. 
By minor changes, that algorithm can be modified to compute the degree 
density function
$p(k) = \frac{Z^*(1,n,k)}{N(1,n)}$ for any given RNA sequence $a_1,\ldots,a_n$.
In the case of a homopolymer, any two positions are allowed to base-pair
(regardless of whether the base pair is a Watson-Crick or wobble pair), 
provided only that every hairpin loop contains at least $\theta$ unpaired
positions. For homopolymers, we have a faster algorithm that requires
$O(K^2 n^3)$ time and $O(Kn^2)$ space. Since nucleotide identity is unimportant,
instead of $Z(i,j,k,h,v)$, we describe the function $\ZZ(m,k,h,v)$, where
$m$ corresponds to the length $j-i+1$ of interval $[i,j]$.

\begin{align*}
\ZZ^*(m,k) &= \sum_{h=0}^{m-\theta-1} \sum_{v=0}^{\theta+1} \ZZ(m,k,h,v) \\
N(m) &= \sum_{k=1}^{\frac{(m-\theta)(m-\theta-1)}{2}} \ZZ^*(m,k) 
\end{align*}

We begin by initializing $\ZZ(m,k,h,v)=0$ for all 
$1 \leq m \leq n$, $0 \leq k \leq \frac{(m-\theta)(m-\theta-1)}{2}$,
$0 \leq h \leq m-2$, and $0 \leq v \leq \theta+1$. If $h<0$, we
assume that $\ZZ(m,k,h,v)=0$. 
\medskip

\noindent
{\bf Base Case A:} For $1 \leq m \leq \theta+1$, define \hfill\break
\begin{align}
\label{eqn:baseAbis}
\ZZ(m,0,0,m) &= 1
\end{align}

\noindent
{\bf Base Case B:} For $m = \theta+2$, define  \hfill\break
\begin{align}
\label{eqn:baseBbis}
\ZZ(m,1,0,0) &=1
\end{align}

\noindent
{\bf Base Case C:} For $m = \theta+2$,
define  \hfill\break
\begin{align}
\label{eqn:baseCbis}
\ZZ(m,1,1,\theta+1) &=1
\end{align}

\noindent
{\bf Base Case D:} For all $m > \theta+2$, define
\hfill\break
\begin{align}
\label{eqn:baseDbis}
\ZZ(m,\frac{(m-\theta)(m-\theta-1)}{2},m-\theta-1,\theta+1)  &=1
\end{align}

\noindent
{\bf Inductive Case A:} For $m > \theta+2$ and
$1 \leq k \leq \frac{(m-\theta)(m-\theta-1)}{2}$, define
 \hfill\break
\begin{align}
\label{eqn:indAbis}
\ZZ(m,k,0,0) &+= \ZZ^*(m-2,k-1)
\end{align}

\noindent
{\bf Inductive Case B:} For $m > \theta+2$, 
$1 \leq k < \frac{(m-\theta)(m-\theta-1)}{2}$, and
$0 \leq h \leq m-\theta-1$, define
 \hfill\break
\begin{align}
\label{eqn:indBbis}
\ZZ(m,k,h,0) &+= \sum\limits_{r=2}^{m-\theta-1} \sum\limits_{k_1+k_2 = k-1}
\sum\limits_{w=0}^{\theta+1} \ZZ(r-1,k_1,h-w,w) \cdot \ZZ^*(m-r-1,k_2)
\end{align}
When implemented, this requires a check that $h-w \geq 0$.

\noindent
{\bf Inductive Case C($v$):}
For $v \in \{1,2,\ldots,\theta+1\}$ and $m > \theta+2$, define
\hfill\break
\begin{align}
\label{eqn:indCvbis}
\ZZ(m,k,h,v) &+= \ZZ^*(m-v-2,k-1-vh) \\
\nonumber
&+ \sum\limits_{r=2}^{m-v-\theta-1}
\sum\limits_{k_1+k_2=(k-1-v h)} \qquad
\sum\limits_{w=0}^{\theta+1} \ZZ(r-1,k_1,h-w,w) \cdot
\ZZ^*(m-v-r-1,k_2)
\end{align}
\medskip

\noindent
{\bf Inductive Case D:} 
For $m > \theta+2$,
$1 \leq k < \frac{(m-\theta)(m-\theta-1)}{2}$, and
$1 \leq h < m-\theta-1$,
\hfill\break
\begin{align}
\label{eqn:indDbis}
\ZZ(m,k,h,\theta+1) &+= \sum\limits_{w=1}^{\theta+1}
\ZZ(m-\theta-1,k- \frac{w(w+1)}{2} - (\theta+1)\cdot(h-w), h-w,w)
\end{align}
Note that $h$ is strictly less than $m-\theta-1$, since the case
$h=m-\theta-1$ occurs only when additionally $v=\theta+1$, which only arises
in the empty structure. The general case for the empty structure was handled
in Base Case D.
When implemented, this requires a check that $h-w \geq 0$.

\section{Statistical methods}
\label{section:statisticalMethods}

Current software for probability distribution fitting of connectivity data, 
such as {\tt Matlab}\texttrademark, {\tt Mathematica}\texttrademark, {\tt R} and 
{\tt powerlaw} \cite{Alstott.po14}, appear to require an input file containing
the connectivity of each node in the network. In the case of RNA secondary structures,
this is only possible for very small sequence length. To analyze connectivity
data computed by the algorithm of Section~\ref{section:fasterAlgoHomopolymer},
we had to implement code to compute the maximum likelihood estimation for
scaling factor $\alpha$ in a power-law fit, the optimal degree 
$k_{\mbox\tiny min}$ beyond which connectivity data is fit by a power-law,
and the associated $p$-value for Kolmogorov-Smirnov goodness-of-fit, 
as described in \cite{newmanPowerLawFitEmpiricalData}. We call the resulting
code {\tt RNApowerlaw}. This section explains those details.

Recall the definition of the {\em zeta function}
\begin{align}
\zeta(\alpha) &= \sum_{n = n_0}^{\infty} n^{-\alpha}
\end{align}
We use both the generalized zeta function (\ref{eqn:genZetaFun}), 
as well as the truncated generalized zeta function (\ref{eqn:genTruncZetaFun}),
defined respectively by
\begin{align}
\label{eqn:genZetaFun}
\zeta(\alpha;n_0) &= \sum_{n = n_0}^{\infty} n^{-\alpha}\\
\label{eqn:genTruncZetaFun}
\zeta(\alpha;n_0,n_1) &= \sum_{n = n_0}^{n_1} n^{-\alpha}
\end{align}
Given a data set $D = \{x_1,\ldots,x_n\}$ of positive integers in the
range $[k_0,k_1]$, the likelihood $L(D|\alpha)$ 
that the data fits a truncated power-law with scaling factor $\alpha$
and range $[k_0,k_1]$ is defined by
\begin{align}
\label{eqn:genLikelihood}
L(D|\alpha) &= \Pi_{i=1}^n \frac{x_i^{-\alpha}}{\zeta(\alpha;k_0,k_1)}
\end{align}
Rather than sampling individual RNA secondary structures to estimate
the connectivity of the secondary structure network for a given homopolymer,
the algorithms from Sections \ref{section:degreeDensityAlgorithm} and
\ref{section:fasterAlgoHomopolymer} directly compute the exact number
$N(k)$ of secondary structures having degree $k$, for all $k$ within
a certain range. It follows that the likelihood $L(D|\alpha)$ that
secondary structure connectivity fits a power-law with scaling factor $\alpha$
is given by
\begin{align}
\label{eqn:likelihoodSecStr}
L(D|\alpha,k_0,k_1) &= \Pi_{k=k_0}^{k_1} \left(
\frac{k^{-\alpha}}{\zeta(\alpha;k_0,k_1)} \right)^{N(k)}
\end{align}
hence the log likelihood is
is given by
\begin{align}
\label{eqn:loglikelihoodSecStr}
\mathcal{L}(D|\alpha,k_0,k_1) &= 
- \left( \log(\zeta(\alpha;k_0,k_1)) \sum_{k=k_0}^{k_1} N(k) \right)
- \left( \alpha \sum_{k=k_0}^{k_1} N(k) \log(k) \right)
\end{align}
The parameter $\widehat{\alpha}$ which maximizes the log likelihood is
determined by applying SciPy function {\tt minimize} (with Nelder-Mead
method) to the negative log
likelihood, starting from initial estimate $\alpha_0$, 
taken from equation (3.7) of \cite{newmanPowerLawFitEmpiricalData}
\begin{align}
\label{eqn:alpha0Newman}
\alpha_0 &= 1 + n \left( 
\sum_{i=1}^n \ln \frac{x_i}{x_{\mbox{\tiny min}}- 1/2}  \right)^{-1}
\end{align}
which in our notation yields
\begin{align}
\label{eqn:alpha0}
\alpha_0 &= 1 + \left( \sum_{k=k_0}^{k_1} N(k) \right) \cdot
\left\{ \sum_{k=k_0}^{k_1} N(k)
\cdot \log \left( \frac{k}{k_0 - 1/2}
 \right)  \right\}^{-1}
\end{align}
In results and tables of this paper, we often write the 
maximum likelihood estimate (MLE) $\widehat{\alpha}$ simply as $\alpha$.

We compute the Kolmogorov-Smirnov (KS) $p$-value, following
\cite{newmanPowerLawFitEmpiricalData}, as follows. 
Given observed relative frequency
distribution $D$ and a power-law fit $P$ with scaling factor $\alpha$, 
the KS distance is defined to
be the maximum, taken over all $k \in [k_0,k_1]$ of the absolute difference
between the cumulative distribution function (CDF) for the data
evaluated at $k$, and the CDF for the power-law, evaluated at $k$
\begin{align}
\label{eqn:KSdistance}
KS(k_{\mbox\tiny min},k_{\mbox\tiny max}) &= 
\max_{k_{\mbox\tiny min} \leq x \leq k_{\mbox\tiny max}} |C_a(x) - C_f(x)|
\end{align}
where $C_a$ and $C_f$ are the actual and fitted cumulative density functions, 
respectively.
The KS $p$-value for the fit of data $D$ by power-law $P$ with scaling
factor $\alpha$, is determined by (1) sampling a large number 
($N=1000$) of synthetic data sets $D_i$ from a true power-law distribution with
scaling factor $\alpha$, (2) computing the KS distance between
each synthetic data set $D_i$ and its power law fit with MLE scaling factor 
$\alpha_i$, (3) reporting the  proportion of KS distances that exceed the
KS distance between the original observed data set and its power-law fit
with scaling factor $\alpha$.

Following \cite{newmanPowerLawFitEmpiricalData}, $k_{\mbox{\tiny min}}$  is
chosen to be that degree $k_0$, such that the KS distance for the optimal
power-law fit is smallest. In contrast, $k_{\mbox\tiny max}$ is always taken
to be the maximum degree in the input data. We have 
implemented Python code to compute $\alpha_0$, $\alpha$, $k_{\mbox\tiny min}$,
KS distance, $p$-value, etc. as described above. In
Section~\ref{section:results}, we compare results of our code with that from
{\tt powerlaw} \cite{Alstott.po14} for very small homopolymers. Though
our code does not do lognormal fits, this is performed by {\tt powerlaw},
where the density function for the lognormal distribution with parameters
$\mu,\sigma$ is defined by
\begin{align}
\label{eqn:lognormal}
p(x) &= \frac{ \exp\left( -\frac{(log(x)-\mu)^2}{2 \sigma^2} \right) }
{x \cdot \sqrt{2 \pi \sigma^2}} 
\end{align}
In computing the $p$-value for power-law goodness-of-fit using
Kolmogorov-Smirnov statistics, it is necessary to sample synthetic data
from a (discrete) power-law distribution with scaling factor $\alpha$, a
particular type of multinomial distribution.
Given an arbitrary multinomial distribution with probability $p_i$ for each
$1\leq i \leq m$, it is straightforward to create $M$ synthetic data sets, each 
containing $N$ sampled values, in time $O(mNM)$; however, since $M=1000$ and $N$ is the
(exponentially large) number of all secondary structures having degrees in
$[k_{\mbox\tiny min}, k_{\mbox\tiny max}]$, the usual {\em sequential} method
would require prohibitive run time. Instead, we implemented the much faster
{\em conditional} method \cite{fastMultinomialSampling}.  Our goal is to
sample from a multinomial distribution given by
\begin{align}
\displaystyle
Prob\left[ X_1=x_1,X_2=x_2,\ldots,X_{m} \right] &=
\frac{N!}{\prod_{i=1}^{m} x_i!} \prod_{i=1}^{m} p_i^{x_i} 
\end{align}
where $m = k_{\mbox\tiny max}-k_{\mbox\tiny min}+1$ is the number
of degrees in the synthetic data, and in the sample set of size $N$ there
are $x_i$ many occurrences of degree $k_{\mbox\tiny min}+i$. To do this, 
we sample $X_1$ from the binomial distribution of $N$ coin tosses with heads 
probability $p_1$, then $X_2$ from the binomial distribution of $N-x_1$ coin tosses
with heads probability  $\frac{p_2}{1-p_1}$,
then $X_3$ from the binomial distribution of $N-x_1-x_2$ coin tosses
with heads probability  $\frac{p_2}{1-p_1-p_2}$, etc. where each $x_i$ is
determined with the function {\tt binom} from Python {\tt Scipy.stats}.

\section{Results}
\label{section:results}

In Section~\ref{section:dataAnalysis}, we use the algorithms described in previous sections
to compute RNA secondary structure connectivity and determine optimal power-law fits, and
in Section~\ref{section:preferentialAttachment} we show that preferential attachment
holds for the network of RNA structures.

\subsection{Analysis of RNA networks using {\tt RNAdensity} and {\tt RNApowerlaw}}
\label{section:dataAnalysis}

The algorithm {\tt RNAdensity} described in Section~\ref{section:fasterAlgoHomopolymer}
was used to compute absolute and relative degree frequencies for the following
cases:
(1) homopolymers of length $n=10,12,\ldots,40$ with $\theta=3$ for maximum possible
degree upper bound $K= \frac{(n-\theta)(n-\theta-1)}{2}$,
(2) homopolymers of length $n=30,35,\ldots,150$ with $\theta=3$, where
degree upper bound $K=2n$ for $n\in [30,100]$ and $K=n+30$ for $n \in [105,150]$,
(3) homopolymers of length $n=30,35,\ldots,150$ with $\theta=1$, where
degree upper bound $K=2n$ for $n\in [30,100]$ and $K=n+30$ for $n \in [105,150]$.
For small homopolymers of length at most $30$, optima values for
$k_{\mbox\tiny min}$, power-law scaling factor $\alpha$, Kolmogorov-Smirnov distance
were determined using software
{\tt powerlaw} {\tt powerlaw} \cite{Alstott.po14} as well as {\tt RNApowerlaw} from
Section~\ref{section:statisticalMethods}. 
Table~\ref{table:pvaluesPowerLawRNApowerlaw} summarizes these results, which show
the agreement between {\tt powerlaw} and {\tt RNApowerlaw}. Moreover, both
both programs suggest that formal hypothesis testing should reject the null hypothesis
that a power-law distribution fits connectivity data; indeed, {\tt powerlaw} determines
a negative log odds ratio $R$ for the logarithm of power-law likelihood over 
lognormal likelihood, indicating a better fit for the lognormal distribution, and
{\tt RNApowerlaw} determines small $p$-values for Kolmogorov-Smirnov goodness-of-fit
of a power-law distribution. Figure~\ref{fig:100mer}a shows connectivity
density function for a 100-mer, with overlaid Poisson and lognormal distributions --
since Erd{\"o}s-R{\'e}nyi random graphs have a Poisson degree distribution
\cite{barabasiReviewSmallWorld}, it follows that RNA secondary structure networks
are strikingly different than random graphs.
Figure~\ref{fig:100mer}b shows a portion of the power-law fit for degrees in
$[k_{\mbox\tiny min}, k_{\mbox\tiny max}]$, where scaling factor 
$\alpha \approx 7.876$ and $k_{\mbox\tiny min}=83$. Although maximum degree probability
at $k_{\mbox\tiny peak}$ is less than $0.05$ for the raw data, the connectivity
density for $[k_{\mbox\tiny min}, k_{\mbox\tiny max}]$ is normalized, which explains
why the degree probability for $k_{\mbox\tiny min}$ is $\approx 0.08$. Visual inspection
suggests an excellent fit for the power-law distribution, despite a Kolmogorov-Smirnov
$p$-value of $\approx 0$. This apparent contradiction highlights the point of view
taken in this paper -- rather than being take sides in the  Barab{\'a}si-Clauset
polemic mentioned in the introduction, our opinion is that a power-law fit for
biological data can provide valuable insight into the underlying network, even though from
a technical point of view, hypothesis testing may lead to rejection of the power-law
fit. The seemingly good power-law fit for RNA connectivity data indicated in
Figure~\ref{fig:100mer} and other figures not shown here led to the investigation
of preferential attachment described in Section~\ref{section:preferentialAttachment}.

Since {\tt powerlaw} requires input files of (individually observed) connectivity degrees, 
when creating Table~\ref{table:pvaluesPowerLawRNApowerlaw}, we could not
run {\tt powerlaw} for homopolymer length greater than $28$, for which latter the input
file contained $50,642,017$ values.  A potentially attractive  alternative is to 
generate input files consisting of  $N \cdot p(k)$ many occurrences
of the value $k$, where $N=10^2,10^3,\ldots,10^7$ denotes the total number
of samples, and where relative frequency $p(k)$ is the proportion of structures having
degree $k$. However, 
Table~\ref{table:dependenceOfPowerLawSoftwareOnSampleSize} shows that neither 
scaling factor $\alpha$ nor $k_{\mbox\tiny min}$ are correct with this 
alternative approach,  even for small homopolymers
of length 20, 30 and 40. This table justifies the need for our implementation of
{\tt RNApowerlaw} as described in Section \ref{section:statisticalMethods}. 
Table~\ref{table:outputRNAdegreeUpTo150} shows 
maximum likelihood scaling factors $\alpha$ and Kolmogorov-Smirnov $p$-values for
optimal power-law fis of connectivity data for homopolymers of lengths from
$30$ to $150$. 

Figure~\ref{fig:cutoffValueGraph}a shows a scatter plot with regression line for
the {\em cut-off} values $x_c$, defined to be the least value such that the probability
that a secondary structure for length $n$ homopolymer has degree greater
that $x_c$ is at most $0.01$. From this figure, we determined that for homopolymer 
length $n>100$, it more than suffices to take degree upper bound $K=n+30$.
Figure~\ref{fig:cutoffValueGraph}b shows the connectivity degree distribution
for a homopolymer of length $20$, where degree $dg(s)$ is redefined to be the number
of structures $t$ that can be obtained from $s$ by adding, removing, or {\em shifting}
a base pair in $s$. The so-called $MS_2$ move set, consisting of an addition, removal
or shift of a base pair is the default move set used in RNA kinetics software 
{\tt kinfold} \cite{Lorenz.amb11}. Although a dynamic programming algorithm was
described in \cite{Clote.po15} to compute the average $MS_2$ network degree, the
methods of this paper do not easily generalize to $MS_2$ connectivity densities.
Figure~\ref{fig:powerlaw20ntHomopolymerMS2} shows a least-squares regression line
for the log-log density plot for $MS_2$ connectivity (computed by brute-force) for
a homopolymer of length $20$, together with an optimal power-law fit computed by
{\tt RNApowerlaw}. Since there are only $106.633$ secondary structures for the 20-mer 
with $\theta=3$, we ran {\tt powerlaw} on $MS_2$ connectivity data,
which determined $\alpha = 6.84$, $k_{\mbox\tiny xmin}=36$, and a log odds
ratio $R= -2.06$ with $p$-value of $0.248$. Since 
{\tt RNApowerlaw} determined
$\alpha = 6.84$,  $k_{\mbox\tiny xmin}=36$, and a 
Kolmogorov-Smirnov $p$-value of $0.219$, we can {\em not} reject the null hypothesis
that a power-law distribution fits the tail of $MS_2$ connectivity data for
a 20-mer.

\subsection{Preferential attachment of RNA secondary structures}
\label{section:preferentialAttachment}

In this section, we provide computational and theoretical arguments that suggest that
{\em preferential attachment} holds in the homopolymer RNA secondary structure model. 
Before proceeding we recall basic definitions and notation. The notion of
homopolymer secondary structure was defined at the beginning of 
Section~\ref{section:preliminaries}; throughout this section, 
we denote the set of
all secondary structures for a length $n$ homopolymer by $\mathcal{S}_n$.
If $s \in \mathcal{S}_{n}$ and $s' \in \mathcal{S}_{n+1}$, then we say that
$s'$ {\em extends} $s$, and write 
$s \prec s'$, if $s'$ is obtained by either (1) appending unpaired nucleotide $n+1$
to the right of $s$, so that the dot-bracket notation of $s'$ is
$s \bullet$, or (2) adding a base pair $(k,n+1)$ to $s$, where 
$k \in [1,n-\theta]$ is
{\em external} to every base pair of $s$, i.e. it is not the case that $i \leq k \leq j$
for any base pair $(i,j)$ of $s$. Since the seminal papers of
\cite{steinWaterman,nussinovJacobson}, this notion of extension has been used as
the basis of recursive and/or dynamic programming algorithms to
count/enumerate all secondary structures and to
compute minimal free energy structures.

A reasonable approach to establish {\em preferential attachment} in the
context of RNA secondary structures is to show that 
if the degree of $s$ is greater than or equal to the degree of $t$ in the
network $\mathcal{S}_n$, then for most extensions $s'$ of $s$, and $t'$ of
$t$, the degree of $s'$ is greater than or equal to the degree of $t'$ in the
network $\mathcal{S}_{n+1}$. We show that this is indeed the case
for homopolymers of modest length, using
by brute-force, exhaustive computations in this section, 
and we rigorously establish this result
for a relaxation $\mathcal{S}^*_n$ of the secondary structure model in
Appendix \ref{section:mathematicalValidationPreferentialAttachment}.


For fixed homopolymer length $n$, define the set $\mathcal{A}_n$ of 
4-tuples $(s,t,s',t')$ by
\begin{align}
\label{eqn:defPrefAttachA}
\mathcal{A}_n &=
\{ (s,t,s',t') : s,t \in \mathcal{S}_n,  s',t' \in
\mathcal{S}_{n+1}, s \ne t, s \prec s', t \prec t',  dg(s) \geq dg(t)  \} 
\end{align}
A 4-tuple $(s,t,s',t') \in \mathcal{A}_n$ {\em succeeds} in demonstrating preferential
attachment if $dg(s') \geq dg(t')$; otherwise the 4-tuple {\em fails} to 
demonstrate preferential attachment. Let {\sc Succ}$_n$ [resp. {\sc Fail}$_n$] denote
the set of 4-tuples that succeed [resp. fail] to demonstrate preferential attachment,
so that $\mathcal{A}_n = \mbox{\sc Succ}_n \cup \mbox{\sc Fail}_n$ (when $n$ is clear,
we drop the subscripts, and we ambiguously also
use {\sc Succ} and {\sc Fail} to denote the sizes of these sets). Our first
quantification of preferential attachment is given by the proportion 
{\sc Succ}/({\sc Succ}+{\sc Fail}):
\begin{align}
\label{eqn:proportionSuccesses}
P({\mbox\sc Succ}_n)  &=
\frac{|\{ (s,t,s',t') \in \mathcal{A}_n: dg(s')\geq dg(t')\}|}
{|\mathcal{A}_n|}  
\end{align}
Since secondary structures have possibly quite different degrees and numbers of 
extensions, a more accurate measure (in our opinion) of preferential attachment is
given by  $\langle p(s',t'|s,t) \rangle$, defined as follows. For distinct, fixed
structures $s,t \in \mathcal{S}_n$, define
\begin{align}
\label{eqn:condProb_stNotAvg}
p(s',t' | s,t) &= 
P\left( dg(s') \geq dg(t') | dg(s) \geq dg(t), s \prec s', t \prec t' |
dg(s) \geq dg(t) \right) \\
&= \nonumber
\left\{ \begin{array}{ll}
0 &\mbox{if $dg(s)<dg(t)$}\\
\frac{ |\{ (s',t'): s',t' \in \mathcal{S}_{n+1}, s'\ne t',
s \prec s', t \prec t', dg(s') \geq dg(t') \} |}
{ |\{ (s',t'): s',t' \in \mathcal{S}_{n+1}, s'\ne t',
s \prec s', t \prec t' \} |} &\mbox{else}
\end{array} \right. \\
\label{eqn:condProb_st}
\langle p(s',t' | s,t) \rangle 
&= \frac{ \sum_{s,t \in \mathcal{S}_n, s\ne t} p(s',t'|s,t) }
{ |\{ (s,t): s,t \in \mathcal{S}_n, s\ne t, dg(s) \geq dg(t) \} |}
\end{align}

To clarify these definitions, we consider a small example. If $n=5$,
then $\mathcal{S}_n$  consists of the two structures
$\bullet \bullet \bullet \bullet \bullet$, and 
$\op \bullet \bullet \bullet \cp$, while
$\mathcal{S}_{n+1}$  consists of the four structures
$\bullet \bullet \bullet \bullet \bullet \bullet$, 
$\op \bullet \bullet \bullet \bullet \cp$, 
$\bullet \op \bullet \bullet \bullet \cp$, 
$\op \bullet \bullet \bullet \cp \bullet$.
Fix $s$ to be $\op \bullet \bullet \bullet \cp$,
and $t$ to be 
$\bullet \bullet \bullet \bullet \bullet$.
Since the only neighbor of
$s$ is $t$, and vice-versa, it follows that $dg(s)=1=dg(t)$.
By definition, an extension $s'$ of $s$ is obtained either by adding
an unpaired nucleotide to $s$ at position $n+1$, or by adding a base
pair $(k,n+1)$ to $s$, where $k$ is external to all base pairs of $s$.
In the current case, the only possible extension of $s$ is produced by
the former rule, thus obtaining
$s'=\op \bullet \bullet \bullet \cp \bullet$. Note that we do {\em not}
consider the structure
$\bullet \op \bullet \bullet \bullet \cp$ to be an extension of $s$.
In contrast, the structure $t= \bullet \bullet \bullet \bullet \bullet$
has three extensions:
$t'_1= \bullet \bullet \bullet \bullet \bullet \bullet$,
$t'_2= \op \bullet \bullet \bullet \bullet \cp$,
$t'_3= \bullet \op \bullet \bullet \bullet  \cp$, where by definition, 
$t'_4= \op \bullet \bullet \bullet \cp \bullet$ is not considered to be an extension of
$t$.
Clearly $dg(s') = dg(t'_2)$, $dg(s') = dg(t'_3)$, but 
$dg(s')=1 \not\geq dg(t'_1)=3$, so
\begin{align*}
\frac{2}{3} &= 
\frac{ |\{ (s',t'): dg(s') \geq dg(t') \land
s \prec s', t \prec t', s,t \in \mathcal{S}_{n+1} \} |}
{ |\{ (s',t'): s \prec s', t \prec t', s,t \in \mathcal{S}_{n+1} \} |}
\end{align*}
so $p(s',t'|s,t)= 0.6667$. If we now take
$s = \bullet \bullet \bullet \bullet \bullet$, and 
$t = \op \bullet \bullet \bullet \cp$, we find that
\begin{align*}
\frac{3}{3} &= 
\frac{ |\{ (s',t'): dg(s') \geq dg(t'),
s \prec s', t \prec t', s,t \in \mathcal{S}_{n+1} \} |}
{ |\{ (s',t'): s \prec s', t \prec t', s,t \in \mathcal{S}_{n+1} \} |}
\end{align*}
so $p(s',t'|s,t) = 1$. The (arithmetical) average of $1$ and $2/3$ is 
$\frac{2+3}{3} = 5/6 = 0.8333$,
which is the value $\langle p(s',t'|s,t) \rangle$ found in the first row
and last column of Table~\ref{table:preferentialAttachment}. In contrast
to this value, averaged over all pairs $s,t \in \mathcal{S}_n$ for which 
$dg(s) \geq dg(t)$, the total number of {\em successes} [resp. 
{\em failures}] is $5$ [resp. $1$], where a success [resp. failure] is
defined as a 4-tuple $(s,t,s',t')$ for which
$s,t \in \mathcal{S}_n$, $s',t' \in \mathcal{S}_{n+1}$,
$s \prec s'$, $t \prec t'$, $dg(s) \geq dg(t)$ and $dg(s') \geq dg(t')$
[resp. $dg(s')<dg(t')$]. Thus we find the value $5/6 = 0.8333$ in
the first row and $7$th column; however, it is not generally true that
{\sc Succ}$_n$/ ({\sc Succ}$_n$+ {\sc Fail}$_n$) agrees with
$\langle p(s',t'|s,t) \rangle$, since $s,t$ may have different degrees  in
$\mathcal{S}_n$, and each may have a different number of 
extensions $s\prec s'$, $t \prec t'$, and each $s',t'$ may each have different
degrees in  $\mathcal{S}_{n+1}$.

For homopolymers of length $5$ to $18$, Table~\ref{table:preferentialAttachment} shows 
the proportion of successes, $P(\mbox{\sc Succ})$, defined
in equation~(\ref{eqn:proportionSuccesses}), as well as the
average preferential attachment probabilities 
$\langle p(s',t'|s,t) \rangle$, defined in  equation (\ref{eqn:condProb_st}).
Values in this table, produced by brute-force, exhaustive computation, were obtained
for each homopolymer length $n \in [5,19]$, by first generating the collections 
$\mathcal{S}_n$, then computing the degrees $dg(s)$ for $s \in \mathcal{S}_n$
by brute force, then considering all ${n \choose 2}$  unordered pairs $s,t$ of distinct 
structures in $\mathcal{S}_n$. So far, the number of instances to consider is large -- 
for instance, when $n=18$, there are ${n \choose 2} = 274,564,461$ unordered
pairs of distinct structures from $\mathcal{S}_n$.
For each pair of distinct structures $s,t$ from $\mathcal{S}_{n}$ that satisfy
$dg(s) \geq dg(t)$, a list $L_s$ [resp. $L_t$] of extensions $s \prec s'$ 
[resp. $t \prec t'$] were computed, where the size of each list is
one plus the number of
positions in $[1,n-\theta]$ that are external to every base pair of
$s$ [resp. $t$]. Subsequently, the proportion of extension pairs $s',t'$
that satisfy $dg(s') \geq dg(t')$ is determined, thus yielding $p(s',t'|s,t)$.
Finally, the mean and standard deviation of the latter yields $\langle p(s',t'|s,t) \rangle$, shown in the last column of the table.
For $n=18$, more than one trillion
($1.36 \cdot 10^9$) 4-tuples $(s,t,s',t')$ where considered for which
$dg(s) \geq dg(t)$ -- this value is used in the denominator of
equation~(\ref{eqn:condProb_st})!

From the values in Table~\ref{table:preferentialAttachment},
it appears that the RNA homopolymer secondary
structure model does demonstrate preferential attachment. This, in our opinion,
may provide theoretical justification for the close approximation of the tail of
degree distributions by a power-law distribution, even though a rigorous statistical
test by bootstrapping Kolmogorov-Smirnov values appears to reject this hypothesis.

\section{Conclusion}
\label{section:conclusion}

Since the pioneering work of Zipf on the scale-free nature of natural
languages \cite{zipf}, various groups have found scale-free networks
in diverse domains ranging from
communication patterns of dolphins \cite{McCowan.jcp02},
metabolic networks \cite{Jeong.n00},
protein-protein interaction networks
\cite{Ito.pnas00,schwikowskiProteinProteinNetworks},
protein folding networks \cite{Bowman.pnas10},
genetic interaction networks \cite{Tong.s04,VanNoort.er04}
to multifractal time series \cite{Budroni.pre17}.
These discoveries have
galvanized efforts to understand biological networks from a mathematical 
and topological standpoint. Using mathematical analysis,
Barab{\'a}si and Albert \cite{Barabasi.s99} established that scale-free
networks naturally emerge when networks are dynamic, whereby newly
accrued nodes are preferentially connected to nodes already having high
degree. On such grounds, one might argue that protein folding 
networks and protein-protein interaction (PPI) networks should
exhibit scale-free properties, since nature is likely to reuse and 
amplify fast-folding domains -- cf.
Gilbert's exon shuffling
hypothesis \cite{Gilbert.n78}.
Indeed, Cancherini et al.  \cite{Cancherini.bg10} have established that
in 4 metazoan species analyzed 
({\em H. sapiens}, {\em M. musculus}, {\em D. , melanogaster}, {\em C. elegans})
those genes, which are enriched in exon shuffling events, 
displayed a higher connectivity
degree on average in protein-protein interaction (PPI) networks; i,e.
such genes had a larger number of interacting partners.
On similar grounds that nature should reuse and amplify successful
metabolic networks, one might argue that metabolic networks should exhibit
scale-free properties. However, rigorous statistical analysis has shown
that metabolic networks fail a goodness-of-fit test for scale-free 
distribution, while PPI satisfy a goodness-of-fit test for scale-free
distributions over a certain range of connectivity
\cite{Khanin.jcb06,newmanPowerLawFitEmpiricalData}.

In this paper, we have introduced a novel algorithm to compute the
connectivity density function for a given RNA homopolymer. Our algorithm
requires $O(K^2n^4)$ run time and $O(Kn^3)$ storage, where $K$ is a 
user-specified
degree bound $K \leq \frac{(n-\theta)(n-\theta-1)}{2}$. Short of exhaustively
listing secondary structures by brute-force, no such algorithm existed 
prior to our work. Since existent software appears unable to perform
power-law fitting for exponentially large RNA connectivity data,
we have implemented code to compute and statistically
evaluate the maximum likelihood power-law fit for an input histogram.
Perhaps this code may prove useful to other groups working with data 
where the underlying data set is so large that it cannot be enumerated,
as is the case with connectivity of RNA secondary structure networks.
Using code {\tt RNAdensity} and {\tt RNApowerlaw}, 
we have computed the connectivity
density function for RNA secondary structure networks for homopolymers
of length up to $150$. Statistical nalysis shows that, almost invariably,
there is no statistically significant power-law fit of connectivity density 
function, despite the fact the strikingly good visual fit shown in
Figure~\ref{fig:100mer} and other data (not shown). Nevertheless, 
power-law fittomg provides a useful paradigm leading to
the establishment of preferential attachment, shown in the previous section 
and Appendix.

%
%

\section*{Acknowledgements}

We would like to thank Amir H. Bayegan for providing the figures in
Section~\ref{section:degreeDensityAlgorithm} and Jenny Baglivo for
a reference for the conditional method to sample from the multinomial distribution.
This work was partially
supported by National Science Foundation grant DBI-1262439.  Any opinions, 
findings, and conclusions or recommendations expressed in this material are
those of the authors and do not necessarily reflect the views of the
National Science Foundation.

\bibliographystyle{plain}

\clearpage

\begin{table}[]
\centering
\caption{Table comparing goodness-of-fit computations for
software {\tt powerlaw} \cite{Alstott.po14} and {\tt RNApowerlaw} for
homopolymer lengths less than 30 nt. Given homopolymer length $n$,
the connectivity density is computed over all secondary structures
for (all possible) degrees $k=1,\ldots,\frac{(n-3)(n-4)}{2}$ using 
the algorithm described in Section~\ref{section:fasterAlgoHomopolymer}. 
Program {\tt powerlaw} requires an input file containing the degrees of
all structures (i.e. containing $S_n$ values, where $S_n$ is the exponentially
large number of all secondary structures), while our program
{\tt RNApowerlaw} requires as input a list of degrees and their (absolute)
frequencies. Table headers as follows: $n$ is homopolymer length,
$S_n$ is the number of all secondary structures,
$\alpha$ is the maximum likelihood value for the scaling factor of the
optimal power-law fit, as computed by {\tt powerlaw} (PL) and
{\tt RNApowerlaw} (RNAPL), KSdist is the Kolmogorov-Smirnov (KS)
distance using equation~(\ref{eqn:KSdistance}), 
$\langle \mbox{KSdist} \rangle$ is the mean KS-distance obtained by
replacing `max' by `mean' in equation~(\ref{eqn:KSdistance}), $R$ is the 
log-odds ratio with associated $p$-value  as computed by {\tt powerlaw},
and the $p$-value in the last column is computed by {\tt RNApowerlaw}
as described in Section~\ref{section:statisticalMethods}. Since
{\tt powerlaw} required more than 24 hours for the computation when 
$n=28$, we did not attempt a computation for $n=30$; in contrast,
{\tt RNApowerlaw} requires a few seconds computation time.  Since
the log-odds ratio $R$ is the logarithm of the power-law likelihood divided
by lognormal likelihood, a negative value $R<0$ indicates that the lognormal
distribution is a better fit for the tail of RNA
secondary structure connectivity data. A small $p$-value computed by
{\tt RNApowerlaw} indicates that RNA connectivity data is not well-approximated
by a power-law distribution. Nevertheless, we believe that the power-law
paradigm provides some valuable insight, given small mean KS-distance
and the fact that preferential attachment could be shown for the network of
secondary structures -- see Section~\ref{section:preferentialAttachment}.
}
\label{table:pvaluesPowerLawRNApowerlaw}
\medskip
\begin{small}
\resizebox{\linewidth}{!}{%
\begin{tabular}{| rrrrrrrrrrr |}
\hline
$n$	&	$S_n$	&	$k_{\mbox\tiny min}$	&	 $\alpha$ (PL)	&	$\alpha$ (RNAPL)	&	KSdist (PL)	&	 KSdist (RNAPL)	&	$\langle \mbox{KSdist} \rangle$	&	log odds ratio R (PL)	&	p-val for R (PL)	&	p-val (RNAPL)	\\
\hline
10	&	65	&	3	&	3.13752	&	3.13753	&	0.05576	&	0.05576	&	0.02721	&	-0.15	&	0.765	&	0.813	\\
12	&	274	&	4	&	3.23011	&	3.23011	&	0.03650	&	0.03650	&	0.01277	&	-0.81	&	0.482	&	0.746	\\
14	&	1184	&	5	&	3.38933	&	3.38935	&	0.02021	&	0.02021	&	0.00669	&	-1.70	&	0.270	&	0.699	\\
16	&	5223	&	6	&	3.51285	&	3.51289	&	0.02252	&	0.02253	&	0.00603	&	-6.78	&	0.029	&	0.051	\\
18	&	23434	&	9	&	3.79069	&	3.79073	&	0.02333	&	0.02333	&	0.00624	&	-16.00	&	0.001	&	0.001	\\
20	&	106633	&	10	&	3.87168	&	3.87165	&	0.02116	&	0.02116	&	0.00581	&	-82.12	&	0.000	&	0.000	\\
22	&	490999	&	10	&	3.85806	&	3.85809	&	0.02304	&	0.02304	&	0.00523	&	-670.64	&	0.000	&	0.000	\\
24	&	2283701	&	14	&	4.16480	&	4.16477	&	0.02242	&	0.02242	&	0.00484	&	-1452.24	&	0.000	&	0.000	\\
26	&	10713941	&	15	&	4.24485	&	4.24486	&	0.02298	&	0.02298	&	0.00417	&	-7129.42	&	0.000	&	0.000	\\
28	&	50642017	&	16	&	4.33086	&	4.33089	&	0.02167	&	0.02168	&	0.00347	&	-33020.89	&	0.000	&	0.000	\\
30	&	240944076	&	\textemdash	&	\textemdash	&	4.33681	&\textemdash&	0.02393	&	0.00298	&\textemdash	&\textemdash&	0.000	\\
\hline
\end{tabular}}
\end{small}
\end{table}

\begin{table}[]
\centering
\caption{Table showing that approximate [resp. exact] scaling factor
$\alpha_0$ [resp. $\alpha$] and minimum degree $k_{\mbox \tiny min}$
for optimal power-law fit of homopolymer 
connectivity data can not be reliably computed by using
software {\tt powerlaw} \cite{Alstott.po14} on data sampled from
relative frequencies.   Approximate value $\alpha_0$  is computed 
from equation~(\ref{eqn:alpha0Newman}), while $\alpha$ is the maximum
likelihood estimate (MLE) of the optimal power-law scaling factor.
Given homopolymer length $n=20,30,40$,
connectivity density is computed over all secondary structures
for (all possible) degrees $k=1,\ldots,\frac{(n-3)(n-4)}{2}$ using 
the algorithm described in Section~\ref{section:fasterAlgoHomopolymer}. 
Since {\tt powerlaw} 
requires input files of (individually observed) connectivity degrees, 
rather than a histogram of (absolute) frequencies $F(k)$ of connectivity
degrees, we generated a file consisting of  $N \cdot p(k)$ many occurrences
of the value $k$, where $N=10^2,10^3,\ldots,10^7$ denotes the total number
of samples, and where relative frequency $p(k)$ is defined by
$p(k) = F(k)/\sum_{k=1}^{(n-3)(n-4)/2} F(k)$. 
In contrast to {\tt powerlaw}, our program
{\tt RNApowerlaw} (RNAPL) computes {\em exact} values from connectivity
degree (absolute) frequencies.  When using
{\tt powerlaw}, it is clearly necessary to create input files of 
ever-increasing sizes $N$, in order to have more accurate values of
$\alpha_0$, $\alpha$ and $k_{\mbox\tiny min}$. Since the number $S_n$ of
RNA secondary structures is exponential in homopolymer length  $n$,
it rapidly becomes impossible to use {\tt powerlaw} for large RNAs --
for instance, table values for $n=40$ required an overnight run of
{\tt powerlaw}, while
our software returned the exact value within a few seconds.
}
\label{table:dependenceOfPowerLawSoftwareOnSampleSize}
\begin{small}
\resizebox{\linewidth}{!}{%
\begin{tabular}{|l|cccccccr|}
\hline
$N$ &	$10^2$	&	$10^3$	&	$10^4$	&	$10^5$	&	$10^6$	&	$10^7$	&	{\tt RNAPL}	&	$S_n$\\
\hline
$\alpha_0$, $n=20$	&	$6.58318$	&	$3.66505$	&	$3.93389$	&	$3.86017$	&	$3.84749$	&	$3.84657$	&	$3.84648$	&	$106633\approx 1.1 \cdot 10^5$	\\
$k_{\mbox\tiny min}$	&	$10$	&	$7$	&	$10$	&	$10$	&	$10$	&	$10$	&	$10$	&	\textemdash	\\
\hline
$\alpha_0$, $n=30$	&	$5.27581$	&	$4.42183$	&	$4.46307$	&	$4.35008$	&	$4.32651$	&	 $4.32272$	&	$4.32213$	&	$240944076 \approx 2.4 \cdot 10^{8}$	\\
$k_{\mbox\tiny min}$	&	$12$	&	$13$	&	$16$	&	$16$	&	$16$	&	$16$	&	$16$	&	\textemdash	\\
\hline
$\alpha_0$, $n=40$	&	$5.15978$	&	$5.09714$	&	$5.03719$	&	$5.24488$	&	$5.16985$	&	$5.70916$	&	$5.94561$	&	$633180247373\approx 6.3 \cdot 10^{11}$	\\
$k_{\mbox\tiny min}$	&	$15$	&	$19$	&	$23$	&	$29$	&	$29$	&	$42$	&	$49$	&	\textemdash	\\
\hline
\end{tabular}}
\bigskip
\resizebox{\linewidth}{!}{%
\begin{tabular}{|l|cccccccr|}
\hline 
$N$ &	$10^2$	&	$10^3$	&	$10^4$	&	$10^5$	&	$10^6$	&	$10^7$	&	{\tt RNAPL}	&	$S_n$\\
\hline
$\alpha$, $n=20$	&	$6.76575$	&	$3.70988$	&	$3.96139$	&	$3.88570$	&	$3.87271$	&	$3.87180$	&	$3.87165$	&	$106633\approx 1.1 \cdot 10^5$	\\
$k_{\mbox\tiny min}$	&	$10$	&	$7$	&	$10$	&	$10$	&	$10$	&	$10$	&	$10$	&	\textemdash	\\
\hline
$\alpha$, $n=30$	&	$5.33162$	&	$4.44651$	&	$4.47963$	&	$4.36511$	&	$4.34122$	&	$4.33739$	&	$4.33681$	&	$240944076 \approx 2.4 \cdot 10^{8}$	\\
$k_{\mbox\tiny min}$	&	$12$	&	$13$	&	$16$	&	$16$	&	$16$	&	$16$	&	$16$	&	\textemdash	\\
\hline
$\alpha$, $n=40$	&	$5.19197$	&	$5.11604$	&	$5.049419$	&	$5.25365$	&	$5.17824$	&	$5.65206$	&	$5.95033$	&	$633180247373\approx 6.3 \cdot 10^{11}$	\\
$k_{\mbox\tiny min}$	&	$15$	&	$19$	&	$23$	&	$29$	&	$29$	&	$41$	&	$49$	&	\textemdash	\\
\hline
\end{tabular}}
\end{small}
\end{table}

\begin{table}[]
\centering
\caption{Table showing maximum likelihood scaling factors $\alpha$ with
associated $p$ values
for optimal power-law fits of RNA secondary structure connectivity data
for homopolymers of length $n=30$ to $150$. Absolute and relative
connectivity degree frequencies were computed by {\tt RNAdensity} from
Section~\ref{section:fasterAlgoHomopolymer}, while the optimal parameters
$\alpha, k_{\mbox\tiny min}$  and $p$-values were computed by
{\tt RNApowerlaw} from Section~\ref{section:statisticalMethods}.
Column headers are as follows: $n$ is sequence length, $k_{\mbox\tiny max}$
is the degree upper bound $K$ for {\tt RNAdensity}, 
$\%$ of $S_n$ indicates the proportion of all secondary structures having
degree bounded by $K=k_{\mbox\tiny max}$,
$k_{\mbox\tiny peak}$ is the location of the density maximum,
$k_{\mbox\tiny mfe}=\lfloor \frac{n-\theta}{2} \rfloor$ is the degree of
the minimum free energy structure (having largest number of base pairs),
$k_{\mbox\tiny min}$ is the optimal lower bound for a power-law fit,
$\alpha(k_{\mbox\tiny min},k_{\mbox\tiny max})$	is the maximum likelihood
scaling factor for power-law fit, 
$KS(k_{\mbox\tiny min},k_{\mbox\tiny max})$ is the Kolmogorov-Smirnov (KS)
distance between connectivity data and power-law fit,
$p$-val	is goodness-of-fit $p$ value for Kolmogorov-Smirnov statistics, and
$\langle KS \rangle$ is the average KS distance, obtained by replacing
`max' by `mean' in equation~(\ref{eqn:KSdistance}). Although {\tt RNAdensity}
determined absolute and relative degree frequencies for homopolymers of
length 130 and 150, for unexplained reasons the {\tt Scipy.optimize} function
{\tt minimize} did not converge in the 
maximum likelihood computation of $\alpha$.
}
\label{table:outputRNAdegreeUpTo150}
\medskip
\begin{small}
\resizebox{\linewidth}{!}{%
\begin{tabular}{| rrrrrrrrrr |}
\hline
$n$	&	$k_{\mbox\tiny max}$	&	$\%$ of $S_n$	&	$k_{\mbox\tiny peak}$	&	$k_{\mbox\tiny mfe}$	&	$k_{\mbox\tiny min}$	&	$\alpha(k_{\mbox\tiny min},k_{\mbox\tiny max})$	&	$KS(k_{\mbox\tiny min},k_{\mbox\tiny max})$	&	$p$-val	&	$\langle KS \rangle$	\\
\hline
30	&	60	&	0.99886074	&	10	&	13	&	16	&	4.412752307	&	0.025636172	&	0.03149541	&	0.006894691	\\
35	&	70	&	0.99917394	&	12	&	16	&	18	&	4.545722158	&	0.025991642	&	0.029727427	&	0.006009813	\\
40	&	80	&	0.999404339	&	14	&	18	&	23	&	4.897040035	&	0.023835647	&	0.026543112	&	0.005845715	\\
45	&	90	&	0.999562564	&	16	&	21	&	30	&	5.342317642	&	0.021749037	&	0.026034495	&	0.006104977	\\
50	&	100	&	0.9996808	&	18	&	23	&	32	&	5.462330089	&	0.020786348	&	0.02382197	&	0.005145287	\\
55	&	110	&	0.999762012	&	20	&	26	&	39	&	5.848765937	&	0.019749956	&	0.022546932	&	0.00518246	\\
60	&	120	&	0.999823183	&	22	&	28	&	41	&	5.965304744	&	0.018803143	&	0.020882921	&	0.004525872	\\
65	&	130	&	0.999866331	&	24	&	31	&	49	&	6.362319737	&	0.017886705	&	0.020202276	&	0.004522192	\\
70	&	140	&	0.999898961	&	26	&	33	&	52	&	6.521229066	&	0.016897879	&	0.018717457	&	0.004036303	\\
75	&	150	&	0.999923045	&	28	&	36	&	60	&	6.876787811	&	0.016113851	&	0.018129669	&	0.004015537	\\
80	&	160	&	0.999941051	&	31	&	38	&	63	&	7.026510665	&	0.015105392	&	0.016718486	&	0.003597117	\\
85	&	170	&	0.999954575	&	33	&	41	&	67	&	7.212562892	&	0.014349852	&	0.015688377	&	0.003328962	\\
90	&	180	&	0.999964901	&	35	&	43	&	74	&	7.495319334	&	0.013571721	&	0.014903651	&	0.003195372	\\
95	&	190	&	0.999972604	&	37	&	46	&	78	&	7.672099669	&	0.012832491	&	0.013974921	&	0.002961822	\\
100	&	200	&	0.999978707	&	40	&	48	&	83	&	7.876228775	&	0.012134176	&	0.01318324	&	0.002776086	\\
105	&	135	&	0.999388278	&	42	&	51	&	67	&	7.559405648	&	0.023127812	&	0.027682433	&	0.007817632	\\
110	&	140	&	0.999432364	&	44	&	53	&	70	&	7.705752635	&	0.022696603	&	0.026966274	&	0.007443879	\\
115	&	145	&	0.999473643	&	46	&	56	&	73	&	7.850242149	&	0.022277404	&	0.059607021	&	0.031881135	\\
120	&	150	&	0.999512397	&	49	&	58	&	77	&	8.052936897	&	0.021847326	&	0.025865936	&	0.007113075	\\
125	&	155	&	0.999548701	&	51	&	61	&	80	&	8.193141238	&	0.021417417	&	0.02520985	&	0.006770253	\\
130	&	160	&	0.999582464	&	53	&	63	&	84	&	8.389838968	&	0.020977798	&	0.024789337	&	0.006763371	\\
135	&	165	&	0.999613747	&	55	&	66	&	88	&	8.583283462	&	0.020543854	&	0.024364073	&	0.006753744	\\
140	&	170	&	0.99964276	&	58	&	70	&	\textendash	&	\textendash	&	\textendash	&	\textendash	&	\textendash	\\
145	&	175	&	0.999669723	&	60	&	71	&	94	&	8.851385266	&	0.019680276	&	0.023075596	&	0.00609451	\\
150	&	180	&	0.999694756	&	62	&	75	&	\textendash	&	\textendash	&	\textendash	&	\textendash	&	\textendash	\\
\hline
\end{tabular} }
\end{small}
\end{table}

\begin{table}[]
\centering
\caption{Table showing secondary structure
preferential attachment probabilities.
The first two columns contain homopolymer length $n$ and
$n+1$, followed by the number of secondary structures in ${S}_n$ and
${S}_{n+1}$, then the total number of 4-tuples $(s,t,s',t')$ that
succeed in demonstrating [resp. fail to demonstrate] preferential
attachment, denoted by {\sc Succ} [resp. {\sc Fail}]. The next column
contains the proportion  {\sc Succ}/({\sc Succ}+{\sc Fail}) of
4-tuples that demonstrate preferential attachment, defined by
equation~(\ref{eqn:proportionSuccesses}), while the last column
contains the expected preferential attachment 
$\langle p(s',t'|s)\rangle$, defined by equation~(\ref{eqn:condProb_st}).
This expectation is obtained by computing the arithmetical average of the
conditional probabilities $p(s',t'|s,t)$, defined by
$p(s',t'|s,t) = P\left( dg(s') \geq dg(t') | dg(s) \geq dg(t), s \prec s', t \prec t' \right)$.
}
\label{table:preferentialAttachment}
\medskip
\begin{small}
\begin{tabular}{| rrrrrrrr |}
\hline
n	&	n+1	&	$S_{n}$	&	$S_{n+1}$	&	{\sc Succ}	&	{\sc Fail}&	{\sc Succ}/({\sc Succ}+{\sc Fail}	&	$\langle p(s',t'|s,t) \rangle$	\\
\hline
5	&	6	&	2	&	4	&	5	&	1	&	83.33\%	&	$0.8333 \pm 0.1667$	\\
6	&	7	&	4	&	8	&	18	&	8	&	69.23\%	&	$0.7222 \pm 0.4157$	\\
7	&	8	&	8	&	16	&	90	&	37	&	70.87\%	&	$0.7748 \pm 0.3260$	\\
8	&	9	&	16	&	32	&	419	&	131	&	76.18\%	&	$0.8105 \pm 0.2941$	\\
9	&	10	&	32	&	65	&	1,891	&	575	&	76.68\%	&	$0.8122 \pm 0.2887$	\\
10	&	11	&	65	&	133	&	7,883	&	2,498	&	75.94\%	&	$0.8125 \pm 0.2891$	\\
11	&	12	&	133	&	274	&	33,069	&	9,763	&	77.21\%	&	$0.8300 \pm 0.2730$	\\
12	&	13	&	274	&	568	&	142,968	&	40,797	&	77.80\%	&	$0.8322 \pm 0.2709$	\\
13	&	14	&	568	&	1,184	&	621,884	&	171,384	&	78.40\%	&	$0.8366 \pm 0.2646$	\\
14	&	15	&	1,184	&	2,481	&	2,723,993	&	723,887	&	79.00\%	&	$0.8428 \pm 0.2587$	\\
15	&	16	&	2,481	&	5,223	&	12,041,929	&	3,108,978	&	79.48\%	&	$0.8478 \pm 0.2556$	\\
16	&	17	&	5,223	&	11,042	&	53,730,451	&	13,544,005	&	79.87\%	&	$0.8518 \pm 0.2523$	\\
17	&	18	&	11,042	&	23,434	&	241,738,083	&	59,258,399	&	80.31\%	&	$0.8561  \pm 0.2485$	\\
18	&	19	&	23,434	&	49,908	&	1,096,087,115	&	261,730,198	&	80.72\%	&	$0.8598  \pm 0.2455$		\\
\hline
\end{tabular}
\end{small}
\end{table}


\clearpage

\begin{figure}
\centering
\includegraphics[width=\textwidth]{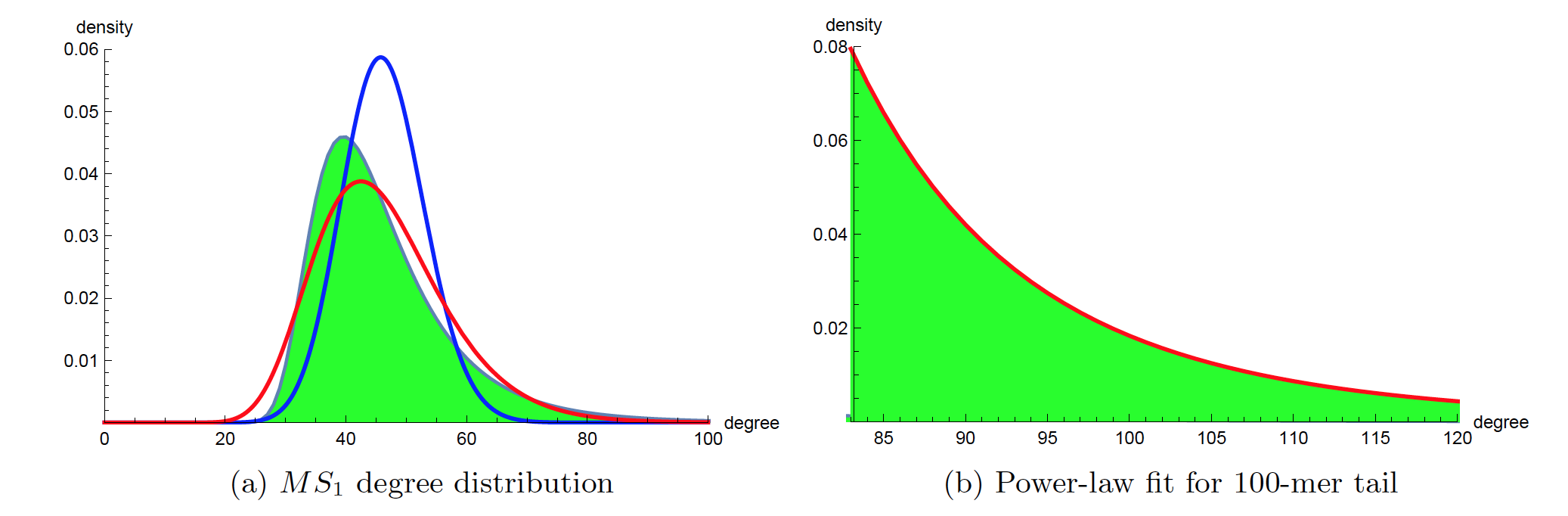}
%
\caption{
(a) Connectivity degree distribution for homopolymer of length $100$ where $\theta=3$, 
computed with the algorithm described in Section~\ref{section:fasterAlgoHomopolymer} for
all degrees bounded by $K = 200$.  There are $6.32 \cdot 10^{32}$ secondary structures
for the 100-mer (exact number $6.31986335936396855341222902079183$), and
$99.9978706904\%$ of the structures have degree bounded by $K$. Using the output
degree densities, the degree mean [standard deviation] is
$\mu=46.2543801196$ [resp. $\sigma=12.2262985078$]; note that the mean computed from 
the algorithm in Section~\ref{section:fasterAlgoHomopolymer} is very close to the
exact degree mean of $\mu=46.2591895818$, computed over all structures using the
different dynamic programming algorithm in \cite{Clote.jcc15}. The 
Poisson distribution (blue curve) with same mean $\mu$ is shown, as well
as the lognormal distribution (red) with parameters $\mu_0=3.80467214577$ and 
$\sigma_0=0.235563374146$; i.e. $\mu_0$ [resp. $\sigma_0$] is the mean
[resp. standard deviation] for logarithms of the connectivity degree -- see
equation~(\ref{eqn:lognormal}). 
(b) Power-law fit of tail
with scaling factor $\alpha=7.8762287746$ and $k_{\mbox\tiny min}=83$,
determined by maximum likelihood. Kolmogorov-Smirnov (KS) distance for the fit is
$0.01213$ -- see  equation~(\ref{eqn:KSdistance}), while average
KS distance for the alpha power-law fit $0.00400$. Nevertheless, since 
the $p$-value $0$ (to 10 decimal places), hypothesis testing would reject the
null hypothesis that the power-law distribution is a good fit for the tail.
}
\label{fig:100mer}
\end{figure}

\begin{figure}
\centering
\includegraphics[width=\textwidth]{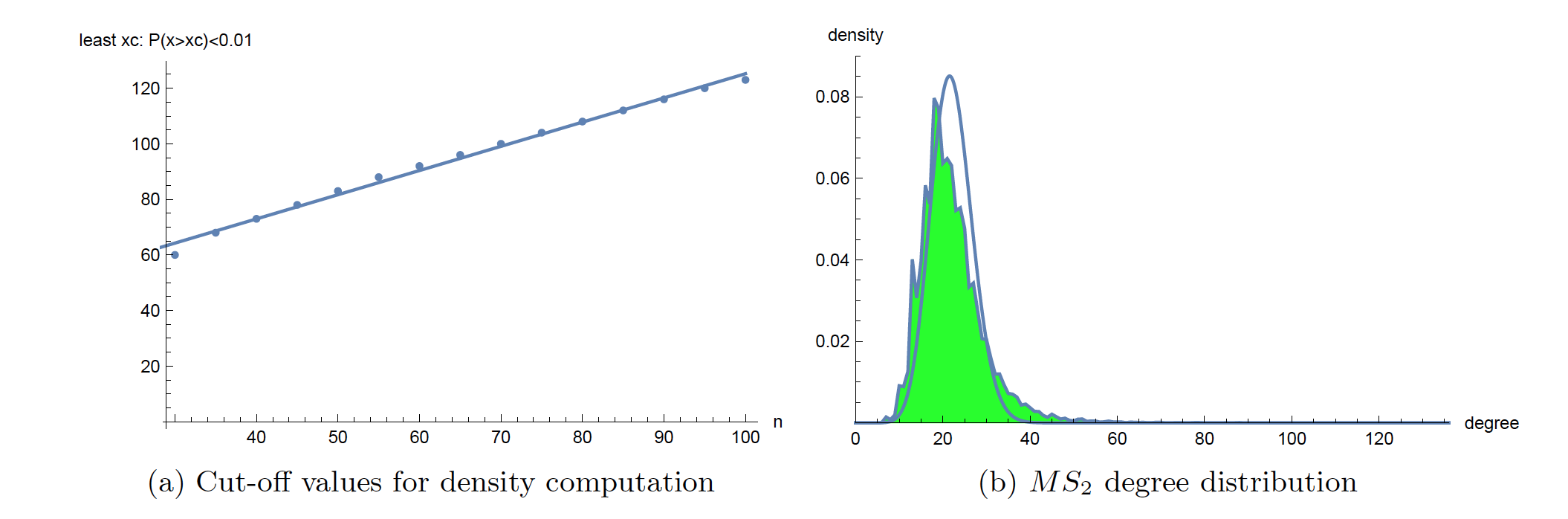}
%
\caption{
(a) Plot of the least cut-off value $x_c$ as a function of homopolymer length
$n$, for $n=30,40,\ldots,100$. 
Here $x_c$ is defined as the least value such that the probability
that a secondary structure for length $n$ homopolymer has degree greater
that $x_c$ is at most $0.01$. For the least-squares fit, the 
regression equation is $y=0.870714 x + 38.1369$, with 
$p$-value of $1.65112 \cdot 10^{-15}$ for slope value, and 
$p$-value of $5.20963 \cdot 10^{-13}$ for the $y$-intercept.
(b) $MS_2$ connectivity for the 106,633 secondary structures for a 
20-nt homopolymer with $\theta=3$ (green shaded curve), with Poisson
distribution of the same mean. Connectivity values range
from $4,\ldots,136$ (with many intermediate gaps before the max degree).
The distribution mean [resp. standard deviation] 
is $\mu=22.0531$ [resp. $\sigma=7.333$]; these values should be contrasted
with the corresponding values of
$\mu'=8.3364$ [resp. $\sigma'=4.7690$] for $MS_1$ connectivity for the 
same 20-nt homopolymer (data not shown).
}
\label{fig:cutoffValueGraph}
\end{figure}

\begin{figure}
\centering
\includegraphics[width=\textwidth]{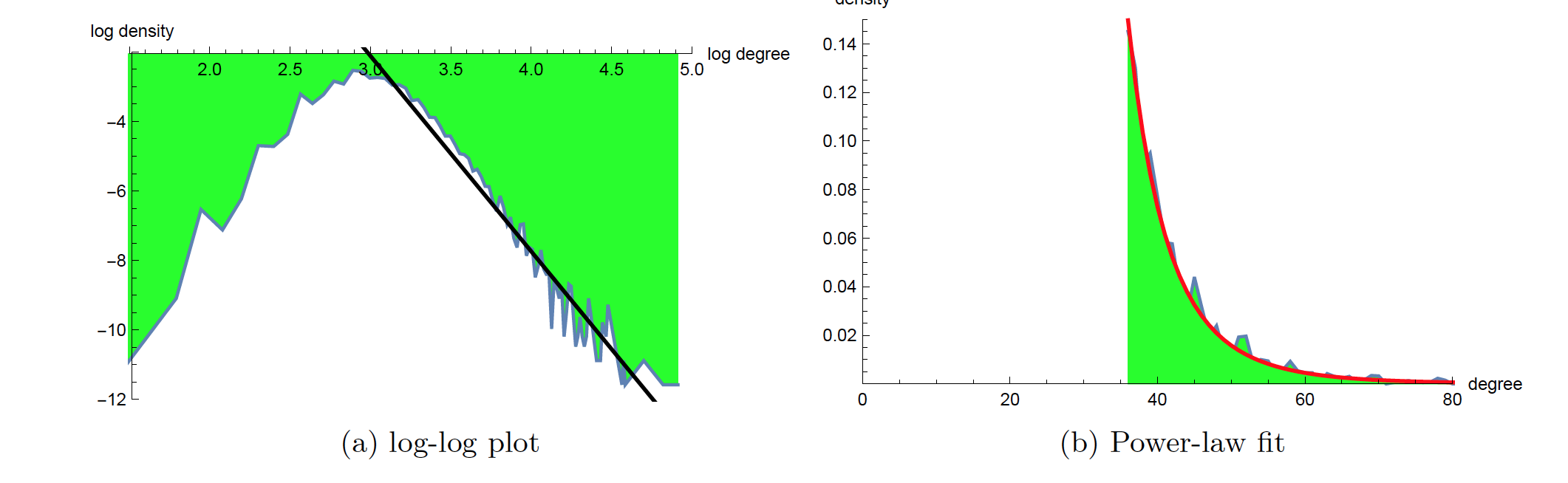}
%
%
\caption{
(a) Plot of $\ln(\mbox{density})$ as a function 
of $\ln(\mbox{degree})$ for the degree distribution for $MS_2$
connectivity of the
20-nt homopolymer with $\theta=3$, for degrees $4,\ldots,136$. 
The distribution tail appears to satisfy
a power-law with exponent $\approx -5.6247$, i.e. 
$p(x) \propto  x^{-5.6247}$,
where $x$ is degree and $p(x)$ is the relative frequency of the number of
nodes having degree $x$ (regression equation log-log plot is
$\ln(p(x)) = 14.7589 - 5.6247 \cdot x$).  
(b) It is well-known that linear regression of the log-log plot is
less reliable than using maximum likelihood when establishing whether the
tail of empirical data is fit by a power-law distribution. For the
$MS_2$ connectivity data of a 20-nt homopolymer, the maximum likelihood
estimation (MLE) of optimal power-law scaling factor is 
$\alpha=6.8257$ with $p$-value is $0.219$ when $k_{\mbox\tiny min}=36$ and
$k_{\mbox\tiny max}=136$. Since the $p$-value is not less than $0.05$,
we can not reject the null hypothesis that $MS_2$ connectivity is well-fit
by a power-law distribution.
}
\label{fig:powerlaw20ntHomopolymerMS2}
\end{figure}

\clearpage
\appendix

\section{Mathematical validation of preferential attachment}
\label{section:mathematicalValidationPreferentialAttachment}

We now proceed to give a rigorous proof of preferential attachment for the
simpler model of {\em pseudo-secondary} structure, in which pseudoknots are {\em allowed}
and $\theta=0$, so that hairpin loops are permitted that contain no unpaired nucleotides. 
Let $\mathcal{S}^*_n$ denote the set of pseudo-secondary structures for a length $n$
homopolymer. By means of an example, when $n=4$, $\mathcal{S}^*_n$ contains the following
nine structures:
$\bullet \bullet \bullet \bullet$,
$\op \cp \bullet \bullet$,
$\op \bullet \cp  \bullet$,
$\op \bullet \bullet \cp$,
$\bullet \op \bullet \cp$,
$\bullet \bullet \op \cp$,
$\op \cp \op \cp$,
$\op \op \cp \cp$,
$\op \ob \cp \cb$.
Only the last structure contains a pseudoknot, for which a distinct type of bracket
must be used. In general, if $s \in \mathcal{S}^*_n$  contains $k$ base pairs,
then $s$ can given by the extended dot-bracket notation over alphabet
$\bullet,a_1,A_1,\ldots,a_k,A_k$, where
symbol $a_i$ [resp. $A_i$] occurs at position $x$ [resp. $y$] if $(x,y)$ is the
$i$th base pair in the lexicographic ordering of base pairs of $s$, while
$\bullet$ occurs at all remaining positions of $s$. 
Throughout the remainder of this
section, structure will mean pseudo-secondary structure.
The following lemma will be used implicitly throughout the remainder of this section
when doing degree computations.

\begin{lemma}
\label{lemma:computationOfDegree}
For any structure $s \in \mathcal{S}^*_n$, the degree of $s$ satisfies
$dg(s) = |s| + {{n-2|s|} \choose 2}$.
\end{lemma}
\begin{proof} 
The first term is due to the fact that $|s|$ structural neighbors of
$s$ can be obtained by removal
of a base pair of $s$.  By adding a base pair $(x,y)$ at any
two of the $(n-2|s|)$ unpaired positions in $s$ we also obtain a neighbor of
$s$. As these are the only neighbors of $s$, the lemma follows. \hfill
\end{proof}

\begin{lemma}
\label{lemma:numBasePairsImpliesDegree}
Let $s,t \in \mathcal{S}^*_n$ be two structures of length $n$. 
If $|s| \leq |t|$ then $dg(s) \geq dg(t)$.
\end{lemma}
\begin{proof}
The proof is now by induction on $|t|-|s|$. In the base case, it
is obvious by the previous lemma that for any $s,t \in \mathcal{S}^*_n$, if $|s|=|t|$ 
then necessarily $dg(s)=dg(t)$. Assume now that $|s| \leq |t|$ and
$|t|-|s|=1$. It follows from the definition of binomial coefficient that
\begin{align*}
{{n-2|t|+2} \choose 2} &= {{n-2|t|+1} \choose 1} + {{n-2|t|+1} \choose 2} \\
&= (n-2|t|+1)  + {{n-2|t|} \choose 1} + {{n-2|t|} \choose 2} \\
&= 2n-4|t|+1 + {{n-2|t|} \choose 2} \\
\end{align*}
We now have 
\begin{align*}
dg(s) &= |s| + {{n-2|s|} \choose 2} = (|t|-1) + {{n-2(|t|-1)} \choose 2}\\
&= (|t|-1) + (2n-4|t|+1) + {{n-2|t|} \choose 2} = 2n-3|t| + {{n-2|t|} \choose 2} \\
dg(t) &= |t| + {{n-2|t|} \choose 2}\\
dg(s)-dg(t) &= 2n-4|t|
\end{align*}
Since $t \in \mathcal{S}^*_n$, clearly $|t| \leq \lfloor n/2 \rfloor$, so
$2n-4|t| \geq 0$, hence $dg(s) \geq dg(t)$. The proof proceeds in a similar 
fashion for larger values of $k=|t|-|s|$ -- in particular, if $|s|=|t|-k$,
then a similar computation shows that
\begin{align*}
dg(s)-dg(t) &= \left( 2kn -4k|t| \right) + \left( \sum\limits_{i=1}^{2k-1} i \right)
- k\\
&= \left( 2kn -4k|t| \right) + 2k(k-1) \geq 2k(k-1)
\end{align*}
The lemma now follows. \hfill
\end{proof}

\begin{corollary}
\label{cor:numBasePairsImpliesDegree}
Let $s,t \in \mathcal{S}^*_n$ be two structures of length $n$. 
Suppose that $|s| < |t|$ and $k=|t|-|s| \geq 1$. Then
$dg(s)>dg(t)$ holds unless $k=1$ and $|t|=n/2$. In the latter case,
$dg(s)=dg(t)$.
\end{corollary}
\begin{proof}
By the proof of the preceding lemma, for $k=|t|-|s|$,
we have $dg(s) - dg(t) \geq 2k(k-1)$, so that
$dg(s)>dg(t)$ for any $k\geq 2$.
If $k=1$ then $dg(s)-dg(t)=2n-4|t|$, which is strictly greater than zero,
unless $n$ is even and $|t|=n/2$. The lemma now follows.  \hfill 
\end{proof}

\begin{lemma}
\label{lemma:degreeImpliesNumBasePairs}
Let $s,t \in \mathcal{S}^*_n$ be two structures of length $n$.
If $dg(s) \geq dg(t)$ then either (1)
$|s| \leq |t|$, or (2) $n$ is even, $|s|=\frac{n}{2}$, $|t|=\frac{n}{2}-1$.
\end{lemma}
\begin{proof}
We begin by a computation.
\begin{align*}
dg(s) \geq dg(t) 
&\Leftrightarrow 
|s|+ { {n-2|s|} \choose 2} \geq |t|+ { {n-2|t|} \choose 2} \\
&\Leftrightarrow 
|s|-|t| \geq  \frac{(n-2|t|)(n-2|t|-1)}{2}- \frac{(n-2|s|)(n-2|s|-1)}{2} \\
&\Leftrightarrow 
2 \cdot ( |s|-|t| ) \geq 
\left( n^2 - 4n|t| + 4|t|^2 -n + 2|t| \right) -
\left( n^2 - 4n|s| + 4|s|^2 -n + 2|s| \right)  \\
&\Leftrightarrow 
2(|s|-|t|) \geq  4(|t|^2-|s|^2) + 4n(|s|-|t|) -2(|s|-|t|) \\
&\Leftrightarrow 
2(|s|-|t|) \geq  -4(|s|-|t|)(|s|+|t|)  +  (|s|-|t|)\left(4n-2 \right)
\end{align*}
If $|s| >|t|$, then by dividing both sides of the last inequality by 
the strictly positive value $2(|s|-|t|)$, we obtain
\begin{align*}
dg(s) \geq dg(t) 
&\Leftrightarrow 
1 \geq -2(|s|+|t|)  + 2n-1
\Leftrightarrow 
2(|s|+|t|) \geq 2n-2
\Leftrightarrow 
|s|+|t| \geq n-1
\end{align*}
Now either $|s| \leq |t|$, which is one of the conclusions of the lemma, or
$|s|>|t|$. In the latter case, then since $|s|,|t| \leq \lfloor \frac{n}{2} \rfloor$,
it must be that $|s|= \lfloor \frac{n}{2} \rfloor$ and
$|t| = \lfloor \frac{n}{2} \rfloor-1$. If $n=2m+1$ is odd, then
$|s|=m$, $|t|=m-1$, so $|s|+|t|=2m-1 < 2m = n-1$. It follows that
$dg(s) \geq dg(t)$ and $|s|>|t|$ can only occur if $n$ is even and
$|s|=\frac{n}{2}$ and $|t|=\frac{n}{2}-1$. This completes the proof of
the lemma. \hfill 
\end{proof}

\begin{corollary}
\label{cor:degreeImpliesNumBasePairs}
Let $s,t \in \mathcal{S}^*_n$ be two structures of length $n$.
If $dg(s) >  dg(t)$ then $|s| < |t|$.
\end{corollary}
\begin{proof}
Assume that $|s|\geq |t|$. Then by
Lemma~\ref{lemma:numBasePairsImpliesDegree}
$|s| \geq |t|$ implies that $dg(s) \leq dg(t)$, which contradicts
the hypothesis of the lemma. It follows that $|s|>|t|$.  \hfill 
\end{proof}

\begin{lemma}
\label{lemma:combinatoricsForFail}
If $n \geq 2$ is an even integer, then
\begin{align}
\label{eqn:FailnComputation}
{\mbox\sc Fail}_n &= 
 \frac{n!}{(n/2)! \cdot 2^{n/2}} \cdot 
\frac{(n-2)!}{((n-2)/2)! \cdot 2^{(n-2)/2}} \cdot
\frac{n(n-1)}{2}  + \\
&\sum\limits_{k=1}^{n/2-1} 
\Big[ {n \choose {2k}} \cdot 
\left( \frac{(2k)!}{k! \cdot 2^{k}} \right) \Big] \cdot
\Big[ {n \choose {2k}} \cdot 
\left( \frac{(2k)!}{k! \cdot 2^{k}} \right) -1 \Big]
\cdot 2k
\end{align}
\end{lemma}
\begin{proof}
Recall that {\sc Fail}$_n$ consists of all 4-tuples $(s,t,s',t')$ such that
$s,t$ are distinct structures in $\mathcal{S}^*_n$, with $dg(s) \geq dg(t)$,
and that $s',t' \in \mathcal{S}^*_{n+1}$ are extensions $s \prec s'$, $t \prec t'$,
but that $dg(s')<dg(t')$. If $s \prec s'$, then either $s'$ is obtained by adding
an unpaired nucleotide at the end of $s$, in which case $|s'|=|s|$, or $s'$ is
obtained by adding a base pair $(k,n+1)$ to $s$, for some $k \in [1,n]$ external to
every base pair of $s$, in which case $|s'|=|s|+1$. It follows that for each
4-tuple $(s,t,s',t')$ in {\sc Fail}$_n$, one of the following cases occurs.
\medskip

\noindent
{\sc Case 1:} $|s|=|s'|$, $|t'|=|t|$.
\hfill\break
Since $(s,t,s',t') \in {\mbox{\sc Fail}}_n$, $dg(s')<dg(t')$, hence by
Corollary \ref{cor:degreeImpliesNumBasePairs}, $|s'|>|t'|$. Since
$|s|=|s'|$ and $|t|=|t'|$, it follows that $|s|>|t|$.
Corollary \ref{cor:numBasePairsImpliesDegree} then implies that 
if $k=|s|-|t|>1$ or $k=1$ and $2|s|<n$, then $dg(s)<dg(t)$, 
a contradiction of the hypothesis that
$dg(s)\geq dg(t)$. It follows that $2|s| \geq n$, and since 
$|s| \leq \lfloor n/2 \rfloor$ and $n$ is even,
it must be that $|s|=n/2$. Now $dg(s)=|s|+ {{n-2|s|} \choose 2} = |s|$.
The only manner in which $dg(s) \geq dg(t)$ is if $|t|=n/2-1$, in which
case $dg(s) = n/2 = dg(t)$. Let
$f_n(1)$ denote the number of 4-tuples in {\sc Fail}$_n$ that satisfy
the hypothesis of the current case. Then
\begin{align}
\label{eqn:fn1}
f_n(1) &= \frac{n!}{(n/2)! \cdot 2^{n/2}} \cdot 
\frac{(n-2)!}{((n-2)/2)! \cdot 2^{(n-2)/2}} \cdot
\frac{n(n-1)}{2}  
\end{align}
Indeed, we claim that the number of $s$ with $|s|=n/2$ for $n$ even is
$\frac{n!}{(n/2)! \cdot 2^{n/2}}$, since 
there are $\frac{n!}{(n/2)! \cdot 2^{n/2}}$ many ways of distributing $n/2$ parentheses:
${n \choose 2}$ choices of the first parenthesis, ${{n-2} \choose 2}$
choices for location of the second parenthesis, etc. However, the parentheses
symbols are indistinguishable, so we then divide by $(n/2)!$. Since $|t|=n/2-1$,
there are $n \choose 2$ choices for where to insert the two unpaired positions;
having fixed the unpaired positions, there are 
$\frac{(n-2)!}{(n/2-1)! \cdot 2^{n/2-1}}$ many ways of filling the remaining
$n-2$ positions with parentheses, accounting for the fact that the parenthesis
symbols are indistinuishable.

For each such $s$, the only extension of $s$ is $s'=s\bullet$; for this $s'$,
$dg(s')=|s'|=|s|=n/2$. 
For each such $t$, there are exactly three possible extensions:
$t'_1=t\bullet$, $t'_2=t \cup \{(x,n+1)\}$, 
$t'_3=t \cup \{(y,n+1)\}$, where positions $x,y$ are unpaired in $t$.
However, only $t'_1$ satisfies $|t'|=t$. Moreover, since $t'_1$ has
three unpaired positions, $dg(t'_1)=|t'_1|+{3 \choose 2}=|t|+3=|s|+3$,
and so $dg(s')<dg(t'_1)$.  It follows that $(s,t,s',t') \not\in \mbox{sc Fail}_n$,
thus justifying equation~(\ref{eqn:fn1}).
\medskip


\noindent
{\sc Case 2:} $|s'|=|s|$, $|t'|=|t|+1$.
\hfill\break
Since $(s,t,s',t') \in {\mbox{\sc Fail}}_n$, $dg(s')<dg(t')$, hence by
Corollary \ref{cor:degreeImpliesNumBasePairs}, $|s'|>|t'|$.  Since
$|s'|=|s|$, and $|t'|=|t|+1$, it follows that $|s|>|t|+1$, hence
$|s|-|t| \geq 2$.
Corollary~\ref{cor:numBasePairsImpliesDegree} now implies that
$dg(s)<dg(t)$, contradicting the hypothesis that $dg(s) \geq dg(t)$.
Consequently, Case 2 contributes no 4-tuple to {\sc Fail}$_n$; however,
{\sc Succ}$_n$ contains all 4-tuples $(s,t,s',t')$ that satisfy t
$dg(s)\geq dg(t)$ as well as the current case assumptions
$|s'|=|s|$, $|t'|=|t|+1$. In particular this includes all 4-tuples for
which $|s|<|t|$, $|s'|=|s|$, and $|t'|=|t|+1$. 
\medskip

\noindent
{\sc Case 3:} $|s'|=|s|+1$, $|t'|=|t|$.
\hfill\break
Note first that since $|s'|=|s|+1$, 
the extension $s'$ is obtained by adding a base pair of
the form $(k,n+1)$ to $s$, where $k \in [1,n]$ is external to all
base pairs of $s$. Now $n$ is even, so it must be that $|s|<n/2$.
Since $(s,t,s',t') \in {\mbox{\sc Fail}}_n$, $dg(s')<dg(t')$, hence by
Corollary \ref{cor:degreeImpliesNumBasePairs}, $|s'|>|t'|$. 
Now $|s'|=|s|+1$, $|t'|=|t|$, 
so it follows that $|s|+1>|t|$, hence $|s| \geq |t|$. 
By hypothesis of the current lemma, $dg(s)\geq dg(t)$, so by
Lemma \ref{lemma:degreeImpliesNumBasePairs}, either $|s| \leq |t|$,
or $n$ is even and $|s|=n/2$, $|t|=n/2-1$. However, we have already
established that $|s|<n/2$, so it must be that $|s|\leq |t|$.
It follows that $|s|=|t|$. 

Since $s,t$ are assumed to be distinct and
$|s|=0=|t|$ implies that both $s,t$ are the empty structure, 
we must have $1 \leq |s|=|t|$. We have already established that
$|s|<n/2$, so if $f_n(3)$ denotes the number of 4-tuples in {\sc Fail}$_n$ 
that satisfy the hypothesis of Case 3, we have
\begin{align}
\label{eqn:fn3}
f_n(3) &=  |\{ (s,t,s',t') : 1 \leq |s|=|t| \leq \frac{n-2}{2}, s\ne t,
dg(s) \geq dg(t), s \prec s', t \prec t', dg(s')<dg(t') \}|\\
\nonumber
&=\sum\limits_{k=1}^{n/2-1} 
\Big[ {n \choose {2k}} \cdot 
\left( \frac{(2k)!}{k! \cdot 2^{k}} \right) \Big] \cdot
\Big[ {n \choose {2k}} \cdot 
\left( \frac{(2k)!}{k! \cdot 2^{k}} \right) -1 \Big]
\cdot 2k
\end{align}
Indeed, for fixed $k$, since $n$ is even, there are
${n \choose {2k}} \cdot \frac{(2k)!}{k! \cdot 2^{k}}$ many choices of structure $s$ 
having $k=|s|$ base pairs. This holds because
there are ${n \choose {2k}}$ ways of choosing $2k$ positions
that will be occupied by $k$ parenthesis symbols. Having selected these
$2k$ positions among positions $[1,n]$, there are
${{2k} \choose 2}$ ways of choosing where to place the first parenthesis pair,
then ${ {2k-2} \choose 2}$  ways of choosing where to place the second parenthesis pair,
etc. and finally, we divide by $k!$ since the parenthesis symbols are indistinuishable.

Since $|t|=|s|$ and $s \ne t$, once $s$ is selected, there is one fewer possibilities
for choice of $t$, hence the number of choices for $t$ is
$\left( {n \choose {2k}} \cdot \frac{(2k)!}{k! \cdot 2^{k}}-1 \right)$.
For fixed $s$ having $k$ unpaired positions, there are $k+1$ possible extensions
$s \prec s'$, and similarly for $t$. Enumerate the extensions of $s$ as
$s'_0,s'_1,\ldots,s'_{2k}$, where $s'_0=s\bullet$, while $s'_1,\ldots,s'_{2k}$
constitute the positions that are paired; similarly enumerate the extensions
of $t$ as $t'_0,t'_1,\ldots,t'_{2k}$.  Now $dg(s)=dg(t)$, since $|s|=|t|$, and
$dg(s')<dg(t')$ holds if and only if $s' \in \{s'_1,\ldots,s_{2k}\}$ and $t'=t'_0$.
For all such choices of $s',t'$ we have $|s'|=|s|+1$ and $t'=|t|$, so the case
hypothesis is satisfied.  This justifies equation~(\ref{eqn:fn3}). 

Since $dg(s') \geq dg(t')$ if and only if $s'=s'_0=s\bullet$, or if
$s' \in \{s'_1,\ldots,s_{2k}\}$ and $t' \in \{t'_1,\ldots,_{2k}\}$, and
for all such choices of $s',t'$ it is not the case that
$|s'|=|s|+1$, $|t'|=|t|$, it follows that there are no 4-tuples satisfying
the current case hypothesis that belong to {\sc Succ}$_n$. 
\medskip

\noindent
{\sc Case 4:} $|s'|=|s|+1$, $|t'|=|t|+1$.
\hfill\break
As in previous cases, $|s'|>|t'|$. Since $|s'|=|s|+1$, $|t'|=|t|+1$, 
it follows that $|s|>|t|$.
Now $|s| \leq \lfloor n/2 \rfloor$, and $n$ is even, so either $2|s|=n$
or $2|s|<n$. If $2|s|=n$, then there are no unpaired positions in $s$,
hence the only extension of $s$ is $s'=s\bullet$, where $|s'|=|s|$.
This is not possible under the hypothesis of the current case. Thus it must
be that $2|s|<n$, hence by Corollary
\ref{cor:numBasePairsImpliesDegree}, $dg(s)<dg(t)$. But this contradicts the
hypothesis that $dg(s) \geq dg(t)$. Subsequently, Case 4 contributes no
4-tuple to {\sc Fail}$_n$.

In contrast, {\em all} 4-tuples $(s,t,s',t')$ that satisfy the hypothesis of
the current case belong to {\sc Succ}$_n$; in particular, if $0\leq |s|<|t|<n/2$
and $|s'|=|s|+1$, $|t'|=|t|+1$, we have $dg(s) \geq dg(t)$ and $dg(s')\geq dg(t')$.

In summary, we have established that
\begin{align*}
{\mbox\sc Fail}_n &= f_n(1)+f_n(3)\\
f_n(1) &= \frac{n!}{(n/2)! \cdot 2^{n/2}} \cdot 
\frac{(n-2)!}{((n-2)/2)! \cdot 2^{(n-2)/2}} \cdot
\frac{n(n-1)}{2}\\
f_n(3) &=\sum\limits_{k=1}^{n/2-1} 
\Big[ {n \choose {2k}} \cdot 
\left( \frac{(2k)!}{k! \cdot 2^{k}} \right) \Big] \cdot
\Big[ {n \choose {2k}} \cdot 
\left( \frac{(2k)!}{k! \cdot 2^{k}} \right) -1 \Big]
\cdot 2k
\end{align*}
This concludes the proof of the lemma.  \hfill
\end{proof}

\begin{lemma}
\label{lemma:combinatoricsForSucc}
If $n \geq 2$ is an even integer, then
\begin{align}
\label{eqn:SuccnComputation}
{\mbox\sc Succ}_n &\geq 
\sum\limits_{k=0}^{n/2} \quad \sum\limits_{\ell=k}^{n/2} 
\left[ {n \choose {2k}} \cdot \frac{(2k)!}{k! \cdot 2^k} \right] \cdot
\left[ {n \choose {2\ell}} \cdot 
\frac{(2\ell)!}{\ell! \cdot 2^{\ell}} 
-1 \right] + \\
\nonumber
&\sum\limits_{k=0}^{n/2-2} \qquad \sum\limits_{\ell=k}^{n/2-1} 
\left[ {n \choose {2k}} \cdot \frac{(2k)!}{k! \cdot 2^k} \right] \cdot
\left[ {n \choose {2\ell}} \cdot 
\frac{(2\ell)!}{\ell! \cdot 2^{\ell}} 
-1 \right] \cdot (n-2\ell) + \\
\nonumber
&\sum\limits_{k=0}^{n/2-1} \quad \sum\limits_{\ell=k+1}^{n/2} 
\left[ {n \choose {2k}} \cdot \frac{(2k)!}{k! \cdot 2^k} \right] \cdot
\left[ {n \choose {2\ell}} \cdot \frac{(2\ell)!}{\ell! \cdot 2^{\ell}} 
\right] \cdot (n-2k) + \\
\nonumber
&\sum\limits_{k=0}^{n/2-1} \quad \sum\limits_{\ell=k}^{n/2-1} 
\left[ {n \choose {2k}} \cdot \frac{(2k)!}{k! \cdot 2^k} \right] \cdot
\left[ {n \choose {2\ell}} \cdot \frac{(2\ell)!}{\ell! \cdot 2^{\ell}} 
-1 \right] \cdot (n-2k)(n-2\ell)
\end{align}
\end{lemma}
\begin{proof}
Recall that {\sc Fail}$_n$ consists of all 4-tuples $(s,t,s',t')$ such that
$s,t$ are distinct structures in $\mathcal{S}^*_n$, with $dg(s) \geq dg(t)$,
and that $s',t' \in \mathcal{S}^*_{n+1}$ are extensions $s \prec s'$, 
$t \prec t'$, for which $dg(s')\geq dg(t')$. 
As in the previous lemma, we consider each of the following four cases.
\medskip

\noindent
{\sc Case 1:} $|s|=|s'|$, $|t'|=|t|$.
\hfill\break
By Lemma~\ref{lemma:numBasePairsImpliesDegree}, if $|s|\leq |t|$
then $dg(s) \geq dg(t)$; moreover, for extensions $s'=s\bullet$ and
$t'=t\bullet$ we have $|s'|=|s|\leq |t|=|t'|$, so $dg(s')\geq dg(t')$.
This justifies the following computation.

\begin{align}
\label{eqn:sn1}
s_n(1) &= \sum\limits_{k=0}^{n/2} \qquad \sum\limits_{\ell=k}^{n/2} 
\left[ {n \choose {2k}} \cdot \frac{(2k)!}{k! \cdot 2^k} \right] \cdot
\left[ {n \choose {2\ell}} \cdot 
\frac{(2\ell)!}{\ell! \cdot 2^{\ell}} 
-1 \right] 
\end{align}
\medskip

\noindent
{\sc Case 2:} $|s'|=|s|$, $|t'|=|t|+1$.
\hfill\break
In the proof of the previous lemma, it was mentioned that under current
case conditions, there are no 4-tuples that belong to
{\sc Fail}$_n$. 
By Lemma~\ref{lemma:numBasePairsImpliesDegree}, if
$0 \leq |s|\leq |t|<n/2$ we have $dg(s)\geq dg(t)$, hence all such 4-tuples
that satisfy current case conditions belong to {\sc Succ}$_n$. 
Noting that there are $(n-2|t|)$ extensions $t'$ obtained by adding
a base pair $(x,n+1)$ to $t$, where $x$ is unpaired in $t$, we obtain
$s_n(2)$ such 4-tuples, where
\begin{align}
\label{eqn:sn2}
s_n(2) &=  \sum\limits_{k=0}^{n/2-2} \qquad \sum\limits_{\ell=k}^{n/2-1} 
\left[ {n \choose {2k}} \cdot \frac{(2k)!}{k! \cdot 2^k} \right] \cdot
\left[ {n \choose {2\ell}} \cdot 
\frac{(2\ell)!}{\ell! \cdot 2^{\ell}} 
-1 \right] \cdot (n-2\ell)
\end{align}
Here we note that the occurrence of $-1$ in 
$\left[ {n \choose {2\ell}} \cdot 
\frac{(2\ell)!}{\ell! \cdot 2^{\ell}} -1 \right]$
is due to the requirement that $s\ne t$.
\medskip

\noindent
{\sc Case 3:} $|s'|=|s|+1$, $|t'|=|t|$.
\hfill\break
For any $0 \leq |s|<|t|<n/2$ Corollary \ref{cor:numBasePairsImpliesDegree} 
implies that $dg(s) > dg(t)$. As well, there are $(n-2|s|)$ many extensions
$s'$ of $s$ obtained by adding a base pair of the form $(x,n+1)$ to $s$,
where $x$ is unpaired in $s$. For each such extension $s'$ and for the
extension $t'=t\bullet$, since
$|s'|\leq |t'|$ we also have $dg(s')\geq dg(t')$. Thus
\begin{align}
\label{eqn:sn3}
s_n(3) &=\sum\limits_{k=0}^{n/2-1} \quad \sum\limits_{\ell=k+1}^{n/2} 
\left[ {n \choose {2k}} \cdot \frac{(2k)!}{k! \cdot 2^k} \right] \cdot
\left[ {n \choose {2\ell}} \cdot \frac{(2\ell)!}{\ell! \cdot 2^{\ell}} 
\right] \cdot (n-2k)
\end{align}
\medskip

\noindent
{\sc Case 4:} $|s'|=|s|+1$, $|t'|=|t|+1$.
\hfill\break
By Lemma~\ref{lemma:numBasePairsImpliesDegree}, if $|s|\leq |t|$ then
$dg(s) \geq dg(t)$. It follows that for any distinct $s,t$ satisfying
$|s| \leq |t|$, for all $n-2|s|$ extensions $s'$ obtained by
adding a base pair of the form $(x,n+1)$ to $s$ where $x$ is unpaired in
$s$, and for all $n-2|t|$ extensions $t'$ obtained by
adding a base pair of the form $(y,n+1)$ to $t$ where $y$ is unpaired in
$t$, we have $dg(s) \geq dg(t)$ and $dg(s') \geq dg(t')$.
Thus
\begin{align}
\label{eqn:sn4}
s_n(4) &=\sum\limits_{k=0}^{n/2-1} \quad \sum\limits_{\ell=k}^{n/2-1} 
\left[ {n \choose {2k}} \cdot \frac{(2k)!}{k! \cdot 2^k} \right] \cdot
\left[ {n \choose {2\ell}} \cdot \frac{(2\ell)!}{\ell! \cdot 2^{\ell}} 
-1 \right] \cdot (n-2k)(n-2\ell)
\end{align}
Note that $-1$ in the expression 
$\left[ {n \choose {2\ell}} \cdot \frac{(2\ell)!}{\ell! \cdot 2^{\ell}} 
-1 \right]$ is due to the requirement that $s \ne t$.

In summary, since we have established that {\sc Succ}$_n$ contains
at least contributions $s_n(1)+s_n(2)+s_n(3)+s_n(4)$, we have
\begin{align*}
{\mbox\sc Succ}_n &\geq s_n(1)+s_n(2)+s_n(3)+s_n(4)\\
s_n(1) &=  \sum\limits_{k=0}^{n/2} \qquad \sum\limits_{\ell=k}^{n/2} 
\left[ {n \choose {2k}} \cdot \frac{(2k)!}{k! \cdot 2^k} \right] \cdot
\left[ {n \choose {2\ell}} \cdot 
\frac{(2\ell)!}{\ell! \cdot 2^{\ell}} 
-1 \right] \\
s_n(2) &=  \sum\limits_{k=0}^{n/2-2} \qquad \sum\limits_{\ell=k}^{n/2-1} 
\left[ {n \choose {2k}} \cdot \frac{(2k)!}{k! \cdot 2^k} \right] \cdot
\left[ {n \choose {2\ell}} \cdot 
\frac{(2\ell)!}{\ell! \cdot 2^{\ell}} 
-1 \right] \cdot (n-2\ell) \\
s_n(3) &=\sum\limits_{k=0}^{n/2-1} \quad \sum\limits_{\ell=k+1}^{n/2} 
\left[ {n \choose {2k}} \cdot \frac{(2k)!}{k! \cdot 2^k} \right] \cdot
\left[ {n \choose {2\ell}} \cdot \frac{(2\ell)!}{\ell! \cdot 2^{\ell}} 
\right] \cdot (n-2k)\\
s_n(4) &=\sum\limits_{k=0}^{n/2-1} \quad \sum\limits_{\ell=k}^{n/2-1} 
\left[ {n \choose {2k}} \cdot \frac{(2k)!}{k! \cdot 2^k} \right] \cdot
\left[ {n \choose {2\ell}} \cdot \frac{(2\ell)!}{\ell! \cdot 2^{\ell}} 
-1 \right] \cdot (n-2k)(n-2\ell)
\end{align*}
This concludes the proof of the lemma.  \hfill
\end{proof}

The computation of $\mbox{\sc Succ}_n$ and
$\mbox{\sc Fail}_n$ for odd integer $n$ is slightly different,
but similar to that of the previous two lemmas. 
Lemmas \ref{lemma:combinatoricsForFail} and
\ref{lemma:combinatoricsForSucc} clearly establish the following
theorem for even $n$, and similar arguments establish the same for
odd $n$.

\begin{theorem}
\label{thm:theoreticalPreferentialAttachment}
For each $n$, $\mbox{\sc Succ}_n/
(\mbox{\sc Succ}_n+ \mbox{\sc Fail}_n) \gg 1/2$.
\end{theorem}
\begin{proof}
We do not carry out the computation using Stirling's factorial approximation,
etc.  since we believe that
little is to be gained by the explict value of this proportion; however,
it suffices to note that the previous two lemmas establish that
$\mbox{\sc Succ}_n \gg \mbox{\sc Fail}_n$.
\end{proof}

\end{document}